\documentclass[a4paper,11pt, notitlepage]{article}
\usepackage{geometry}
\usepackage{graphicx}
\usepackage{setspace}
\usepackage{amsfonts}
\usepackage[bbgreekl]{mathbbol}
\DeclareSymbolFontAlphabet{\mathbb}{AMSb}
\DeclareSymbolFontAlphabet{\mathbbl}{bbold}

\usepackage[english]{babel}
\usepackage{caption}
\usepackage{subcaption} 
\usepackage{mathtools}
\usepackage{mathbbol}
\usepackage{enumerate}
\usepackage{xcolor}
\newcommand{\texorpdfstring}[2]{#1}

\usepackage{array}
\newcommand\Tstrut{\rule{0pt}{3.4ex}}         
\newcommand\Bstrut{\rule[-0.9ex]{0pt}{0pt}}   

\usepackage{amsmath}
\usepackage{graphicx}
\usepackage{graphics}
\usepackage{amssymb}
\usepackage{amsthm}
\usepackage{thmtools, thm-restate}

\newtheorem{theorem}{Theorem}[section]
\newtheorem{definition}[theorem]{Definition}
\newtheorem{lemma}[theorem]{Lemma}
\newtheorem{corollary}[theorem]{Corollary}

\newcommand{\figref}[1]{\figurename~\ref{#1}}

\newcommand{\ceil}[1]{\left\lceil #1 \right\rceil}
\newcommand{\floor}[1]{\left\lfloor #1 \right\rfloor}
\newcommand{\union}{\cup}
\newcommand{\inter}{\cap}

\newcommand{\T}{{t^*}}
\newcommand{\card}[1]{\left|#1\right|}
\newcommand{\G}{\mathcal{G}}

\newcommand{\RT}{\mathcal{T}}
\newcommand{\R}{\mathcal{R}}
\newcommand{\I}[3]{\mathcal{I}_{#1}^{#2}(#3)}
\newcommand{\Out}[3]{\mathcal{O}_{#1}^{#2}(#3)}
\newcommand{\expo}[2]{#1 ^ {(#2)}}
\newcommand{\graph}[3]{#1_{#2}\circ \hdots \circ #1_{#3}}
\newcommand{\N}{\mathbb{N}}
\newcommand{\F}{\mathcal{F}}

\newcommand{\1}{\mathbbl{1}}

\title{Asymptotically Tight Bounds on \\ the Time Complexity of Broadcast and its Variants \\in Dynamic Networks \thanks{This project has received funding from the European Research Council (ERC) under the European Union’s Horizon 2020 research and innovation programme (grant agreement No. 101019564)\includegraphics[scale=0.4]{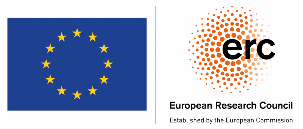}.
This work was further supported by the Austrian Science Fund (FWF) and netIDEE SCIENCE project P 33775-N, as well as the FWF project I 4800-N (ADVISE).}

}
\author{Antoine El-Hayek\\Faculty of Computer Science\\UniVie Doctoral School Computer Science DoCS\\University of Vienna, Austria
\and
Monika Henzinger\\Faculty of Computer Science\\University of Vienna, Austria\\
\and
Stefan Schmid\\TU Berlin, Germany \\ Fraunhofer SIT, Germany
}
\date{}

\begin{document}
\maketitle

\sloppy 
\begin{abstract}
Data dissemination is a fundamental task in distributed computing. This paper studies \emph{broadcast problems} in various innovative models where the communication network connecting $n$ processes is dynamic (e.g., due to mobility or failures) and controlled by an adversary. 

In the first model, the processes transitively communicate their ids in synchronous rounds along a rooted tree given in each round by the adversary whose goal is to maximize the number of rounds until \emph{at least one id is known by all processes}. 
 Previous research has shown a $\lceil{\frac{3n-1}{2}}\rceil-2$ lower bound and an $O(n\log\log n)$ upper bound.
We show the first linear upper bound for this problem, namely $\lceil{(1 + \sqrt 2) n-1}\rceil \approx 2.4n$.

We extend these results to the setting where the adversary gives in each round $k$-disjoint forests and their goal is to maximize the number of rounds until there is a set of $k$ ids such that \emph{each process knows of at least one of them}. We give a $\left\lceil{\frac{3(n-k)}{2}}\right\rceil-1$ lower bound and a $\frac{\pi^2+6}{6}n+1 \approx 2.6n$ upper bound for this problem.

Finally, we study the setting where the adversary gives in each round a directed graph with $k$ roots and their goal is to maximize the number of rounds until \emph{there exist $k$ ids that are known by all processes.}
We give a $\left\lceil{\frac{3(n-3k)}{2}}\right\rceil+2$ lower bound and a $\lceil { (1+\sqrt{2})n}\rceil+k-1 \approx 2.4n+k$ upper bound for this problem.

For the two latter problems no upper or lower bounds were previously known.
\end{abstract}

\section{Introduction}

Data dissemination is one of the most fundamental tasks in distributed systems. 
This paper studies data dissemination in an innovative model where the communication network connecting $n$ processes is \emph{dynamic}. In particular, we consider a worst-case perspective and assume that the information flow between the processes is controlled by an \emph{oblivious message adversary} which may drop an arbitrary set of messages sent by some processes in each round. This results in a sequence of directed communication graphs, whose edges tell which process can successfully send a message to which other process in a given round. 
The oblivious message adversary model is appealing because it is conceptually simple and still provides a highly dynamic network model: The set of allowed graphs can be arbitrary, and the nodes that can communicate with one another can vary greatly from one round to the next. It is, thus, well-suited for settings where significant transient message loss occurs, such as in wireless networks subject to interference, jamming, or mobility.

We look into three data dissemination problems in dynamic networks: \emph{broadcast, cover} and \emph{$k$-broadcast}. These problems come in many flavors and feature intriguing connections to other classic problems such as leader(s) election, regular and $k$-set consensus (also known as $k$-set agreement), for which our problems' time complexity is typically a lower bound.

In particular, we assume that each process has a unique id and
in every message each process communicates all the ids it knows of so far. We first study a fundamental model where communication happens along arbitrary rooted trees, chosen by an adversary who aims to maximize the \emph{broadcast time}, which is the number of rounds it takes until there exists a process that everyone knows of. We then extend our investigations to sparser networks, considering $k$-forests (a union of $k$ rooted trees). Here the adversary will maximize the \emph{cover time}, which is the number of rounds it takes until there exists $k$ process such that everyone knows of at least one of them.
Moreover, in more highly connected networks, we study $k$-rooted dynamic networks (directed graphs with $k$ roots), where the adversary aims at maximizing the \emph{$k$-broadcast time}, which is the number of rounds it takes until there exist $k$ processes that everyone knows of.
Before presenting our results, we introduce our model more formally.

\paragraph*{Model}

Let $n$ be the number of processes and let each process have a unique identifier from $[n]$. 
Let $\G$ be a fixed set of directed networks with $n$ nodes such that each node has a unique identifier from $[n]$. There will be a sequence of rounds $t = 1, 2, \dots, $ such that, in each round $t$, an adversary chooses a network $\tilde G_t$ from $\G$, which determines the communication links in round $t$ as follows:
Initially every process {\em knows} (or {\em has heard}) of its own identifier.
During round $t$ process $i$ sends all identifiers it knows of to its out-neighbors in $\tilde G_t$.
The rounds stop whenever the objective -- broadcast, $k$-broadcast or cover of size $k$ -- is attained.
The goal of the adversary is to maximize the number of rounds.

To model the information propagation we use graph products: 

\begin{definition}
If $A=([n], E_1)$ and $B=([n], E_2)$ are two directed networks on $n$ nodes, then the \emph{product graph} $A\circ B$ is the network on $n$ nodes, with edge set $E$, where $(x,y) \in E$ if and only if there exist a node $z$ such that $(x,z) \in E_1$ and $(z,y) \in E_2$.
\end{definition}

Consider round $t$ and let $G(t)$ be the {\em product graph} $G(t)=G_1\circ \hdots \circ G_t$, where $G_i$ is created from $\tilde G_i$ by adding a self-loop to every node. Note that in $G(t)$ the in-neighbors of a process $x$ are exactly the processes that $x$ knows of after $t$ rounds, and its out-neighbors are the processes it has sent information to. We added a self-loop to every node in $G_i$ to capture the fact that no process ``forgets'' any piece of information in any round. An example is given in \figref{fig:example1}.
\begin{figure}
    \centering
    \includegraphics[width=0.65\linewidth]{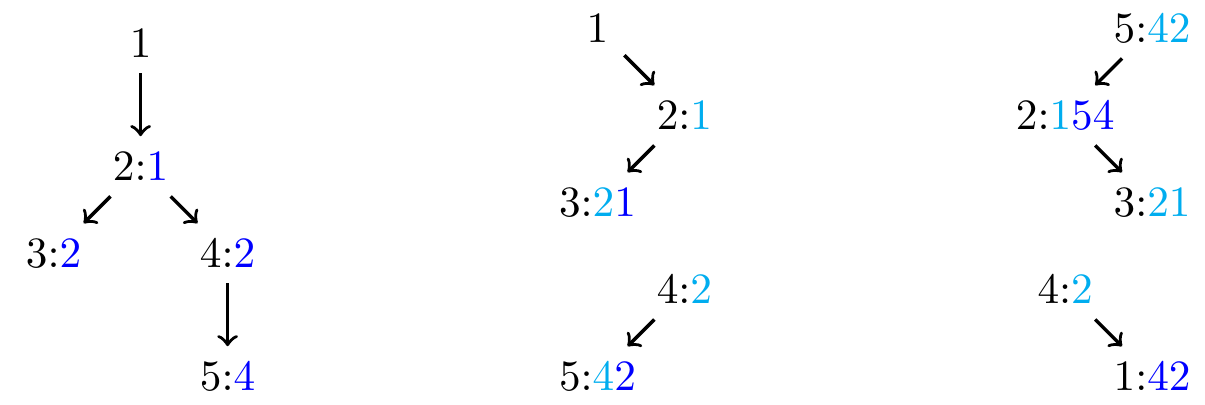}
    \caption{Example of information propagation. Self-loops are omitted, but are present on every node in every round. $x:\textcolor{cyan}{y}\textcolor{blue}{z}$ means that $y$ was an in-neighbor of $x$ before the current round (and still is), whereas $z$ is a new in-neighbor of $x$. We omit mentioning $x$ in its own in-neighbors as it is always the case.}
    \label{fig:example1}
\end{figure}

We will consider three different models, each of which has a different objective, detailed below. \figref{table} summarizes the 3 models, gives examples and states the results.

\paragraph*{Broadcasting on Trees}
\begin{definition}
    Let $\RT_n$ be the set of all rooted trees with a self-loop added at every node.
\end{definition}

In this variant, the networks given by the adversary are restricted to trees in $\RT_n$ and we analyze
the {\em broadcasting time }$\T$, which is the smallest round $t$ such that there exists a node in $G(t)$ with an out-edge to every other node. Note that this corresponds to 
 a process such that every other process has heard of its identifier. We will say that {\em the node has broadcast (its identifier to everyone)} or simply {\em broadcast has happened}. Trees being rooted ensure that broadcast happens in a finite number of rounds\footnote{If the adversary can choose a non-rooted graph, it could repeat this graph indefinitely, preventing broadcast.}.

\begin{definition}
The broadcast time $\T(G_1, G_2, \hdots)$ of a sequence of graphs $G_1, G_2, \hdots$, is
$$
\T(G_1, G_2, \hdots)=\min\{t \in \N: \exists x \in [n], \forall y \in [n], (x,y)\in G_1\circ\hdots \circ G_t\}
$$

\end{definition}


\begin{definition}
The broadcast time $\T(\G)$ of an adversary $\G$, is defined as follows:
$$
\T(\G)=\max\{\T(G_1, G_2, \hdots): \forall i \in \N, G_i \in \G\}
$$
\end{definition}
We will give tight asymptotic bounds on $\T(\RT_n)$.
%
Note that even in the simple case where the adversary gives the same directed tree in each round, the broadcast time can be as large as $n-1$, namely if the tree is simply a path. Conversely, in each round, it is easy to see that at least one new edge appears in the product graph\footnote{In each round, the identifier of the root reaches someone new.}, and thus the broadcast time is at most $n^2$. This raises the question of how large the broadcast time can be made if in each round a different directed tree can be used.

This has been an open question for several years. Results from Charron-Bost and Schiper in 2009~\cite{charron2009heard} and Charron-Bost, F{\"u}gger, and Nowak in 2015~\cite{charron2015approximate} imply an $n \log n$ upper bound. In 2019, Zeiner, Schwarz, and Schmid~\cite{schwarz2017linear} gave a linear upper bound when the adversary is restricted to trees with either a constant number of leaves or a constant number of inner nodes. They also gave a $\ceil{\frac{3n-1}{2}}-2$ lower bound. In 2020, F{\"u}gger, Nowak, and
Winkler~\cite{fugger2020radius} improved the general upper bound to $2n\log\log n + O(n)$. So far, it has been an open conjecture~\cite{schwarz2017linear} whether the broadcast time is linear for arbitrary sequences of rooted trees.

\paragraph*{Covering on $k$-forests}
\begin{definition}
    We define $\F_n^k$ to be the set of all forests over $n$ processes which are the union of $k$ rooted trees and a self-loop is added at every node.
\end{definition}
In this variant, the networks given by the adversary are restricted to networks from $\F_n^k$, and we analyze the
{\em cover time} $t^c_k$, which is the smallest round $t$ such that there exists a set $I\subset [n]$, $\card I \leq k$, such that every node $x$ of $G(t)$ has at least one in-neighbor from $I$. Said differently, there exists a set of $k$ nodes such that every node knows of {\em at least one} of them.

\begin{definition}
The cover time $t^c_k(G_1, G_2, \hdots)$ of a sequence of graphs $G_1, G_2, \hdots$, is 
$$
t^c_k(G_1, G_2, \hdots)=\min\{t \in \N: \exists x_1, \hdots, x_k \in [n], \forall y \in [n], \exists i \in [k], (x_i,y)\in G_1\circ\hdots \circ G_t\}
$$

\end{definition}


\begin{definition}
The cover time $t^c_k(\G)$ of an adversary $\G$, is defined as follows:
$$
t^c_k(\G)=\max\{t^c_k(G_1, G_2, \hdots): \forall i \in \N, G_i \in \G\}
$$
\end{definition}

We will give tight asymptotic bounds on $t^c_k(\F_n^k)$. To the best of our knowledge, there is no prior work that gives upper or lower bounds on $t^c_k(\F_n^k)$.

\paragraph*{$k$-Broadcasting on $k$-rooted Networks}
\begin{definition}
    Let $\R_n^k$ be the set of all (directed) networks over $n$ nodes that
        (1) have $k$ roots, that is $k$ different processes $r_1, \hdots , r_k$ such that there exists a directed path from any $r_i$ to any process in $[n]$, and (2)
        have a self-loop at every node.
\end{definition}
In this variant, the networks given by the adversary are restricted to networks from $\R_n^k$, and we analyze the 
{\em $k$-broadcasting time} $t^*_k$, which is the smallest round $t$ such that there exist $k$ nodes in $G(t)$ with an out-edge to every other node. Said differently, there exists a set of $k$ nodes such that every node knows of {\em all} of them. 
\begin{definition}
The $k$-broadcast time $t^*_k(G_1, G_2, \hdots)$ of a sequence of graphs $G_1, G_2, \hdots$, is  
\begin{multline*}
        t^*_k(G_1, G_2, \hdots)=\min\{t \in \N: \exists x_1, \hdots, x_k \in [n], \forall i \neq j, x_i \neq x_j \\ \land \forall y \in [n], \forall i \in [k], (x_i,y)\in G_1\circ\hdots \circ G_t\}
\end{multline*}
\end{definition}


\begin{definition}
The $k$-broadcast time $t^*_k(\G)$ of an adversary $\G$, is defined as follows:
$$
t^*_k(\G)=\max\{t^*_k(G_1, G_2, \hdots): \forall i \in \N, G_i \in \G\}
$$
\end{definition}

We will give tight asymptotic bounds on $t^*_k(\R_n^k)$. To the best of our knowledge, there is no prior work that gives upper or lower bounds on $t^*_k(\R_n^k)$.

\paragraph*{Contribution}
We give asymptotically tight bounds for all three settings, see also \figref{table}.

(1) First
we settle the open problem about time complexity of broadcast in dynamic
trees, by showing that it is \emph{linear}. Hence, Zeiner et al.'s~\cite{schwarz2017linear}
conjecture is true.
In particular, we present an upper bound of $\ceil { (1+\sqrt{2})n} \approx 2.4n$, which complements their $\left\lceil{\frac{3n-1}{2}}\right\rceil-2$ lower bound. 

(2) We further show that covering on $k$-forests also takes linear time, by giving a $\frac{\pi^2+6}{6}n+1\approx 2.6n$ upper bound and a $\left\lceil{\frac{3(n-k)}{2}}\right\rceil-1$ lower bound.

(3) Finally, we show that $k$-broadcasting on $k$-rooted networks is linear as well, by giving an upper bound of $\ceil { (1+\sqrt{2})n}+k-1 \approx 2.4n+k$, and a lower bound of $\left\lceil{\frac{3(n-3k)}{2}}\right\rceil+2$. 

\begin{figure}[ht]

\begin{center}
\resizebox{\linewidth}{!}{
\begin{tabular}{|c c c c|} 
 \hline
 \rule{0pt}{2.7ex} \Bstrut Model: & Broadcasting on Trees & Covering on $k$-Forests& $k$-Broadcasting on $k$-rooted Networks  \\ [0.5ex] 
 \hline
 Adversary: & \includegraphics[width=0.132\linewidth]{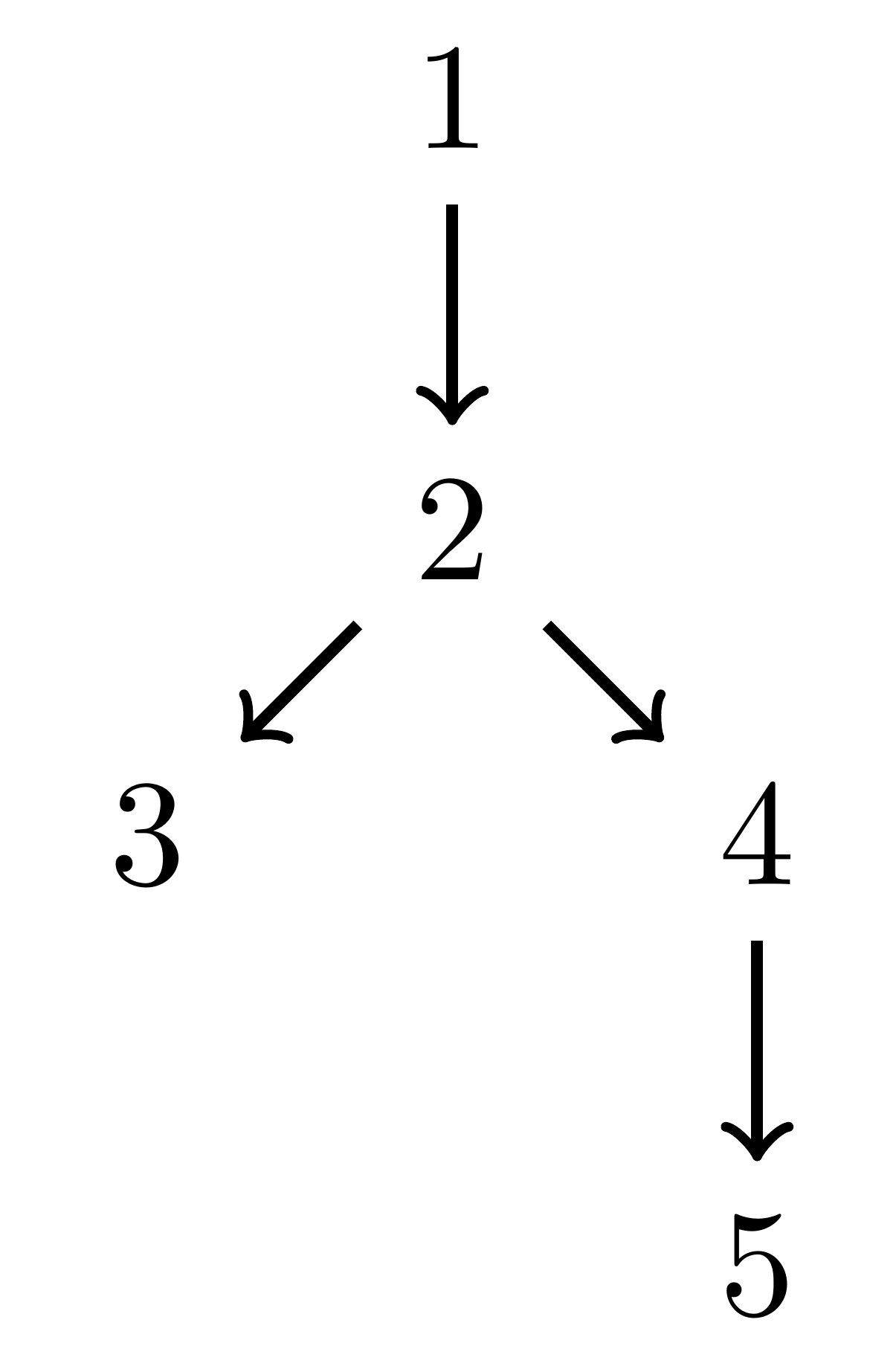} & \includegraphics[width=0.3\linewidth]{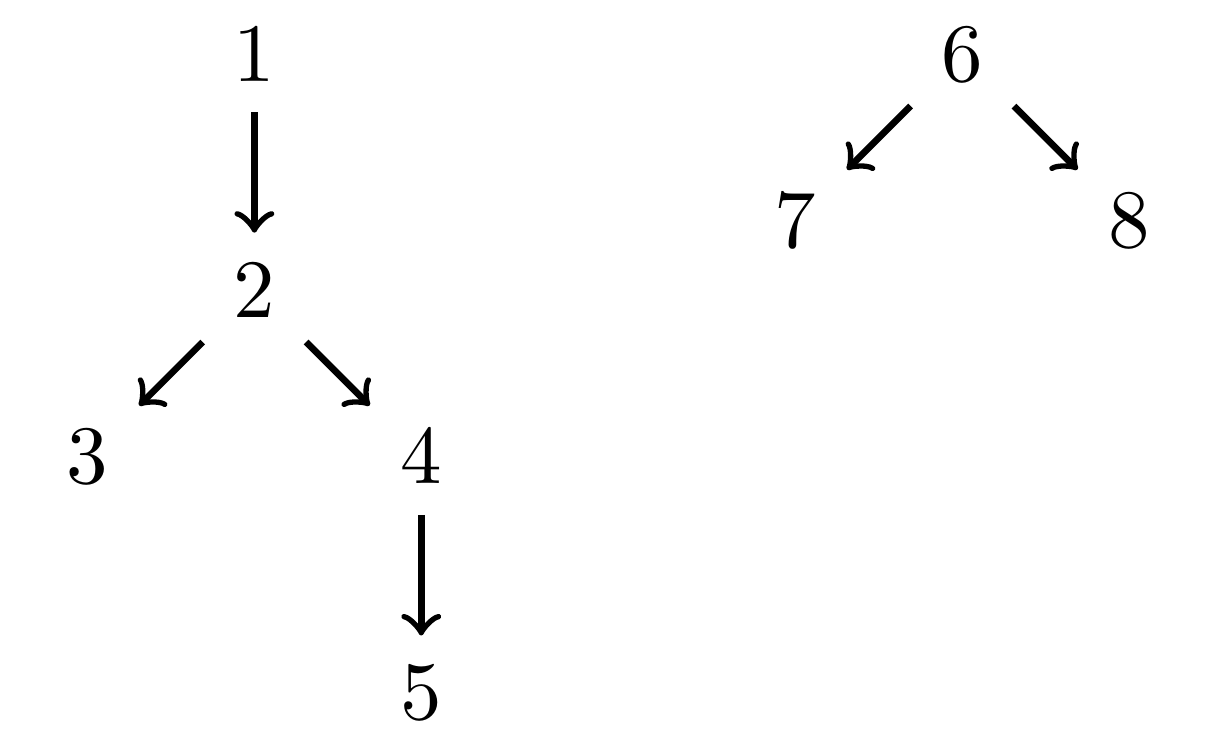}& $\includegraphics[width=0.12\linewidth]{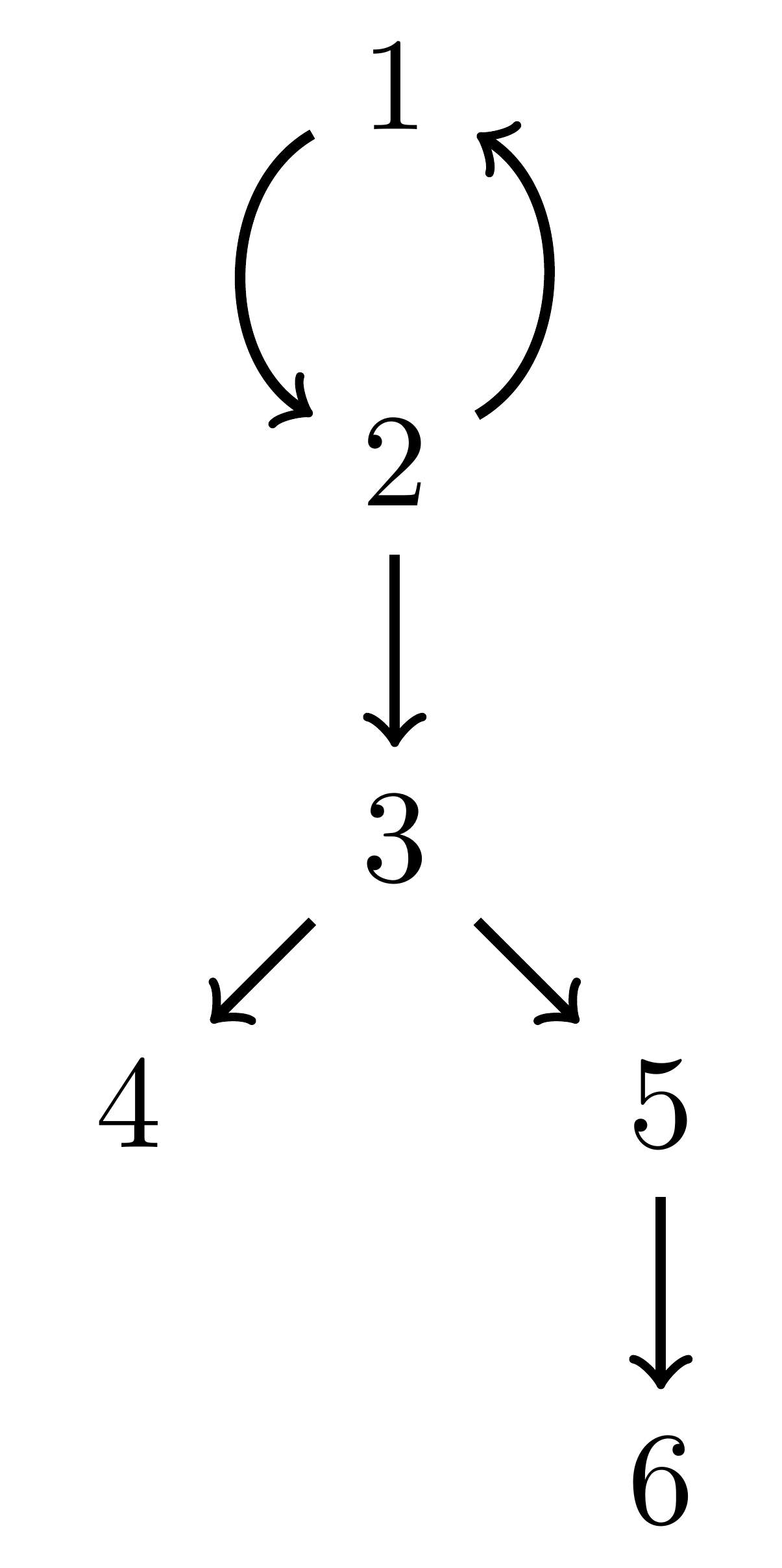}$  \\ 
 
  &Rooted trees over $n$ nodes. &   Forests over $n$ processes,& Networks over $n$ processes and $k$ roots.     \\
  
   & Example for $n=5$ & composed of $k$ rooted trees.& Example for $n=6$ and $k=2$\\
   &&Example for $n=8$ and $k=2$&\\
   [0.5ex] 
   \hline
   Objective:&\includegraphics[width=0.1\linewidth]{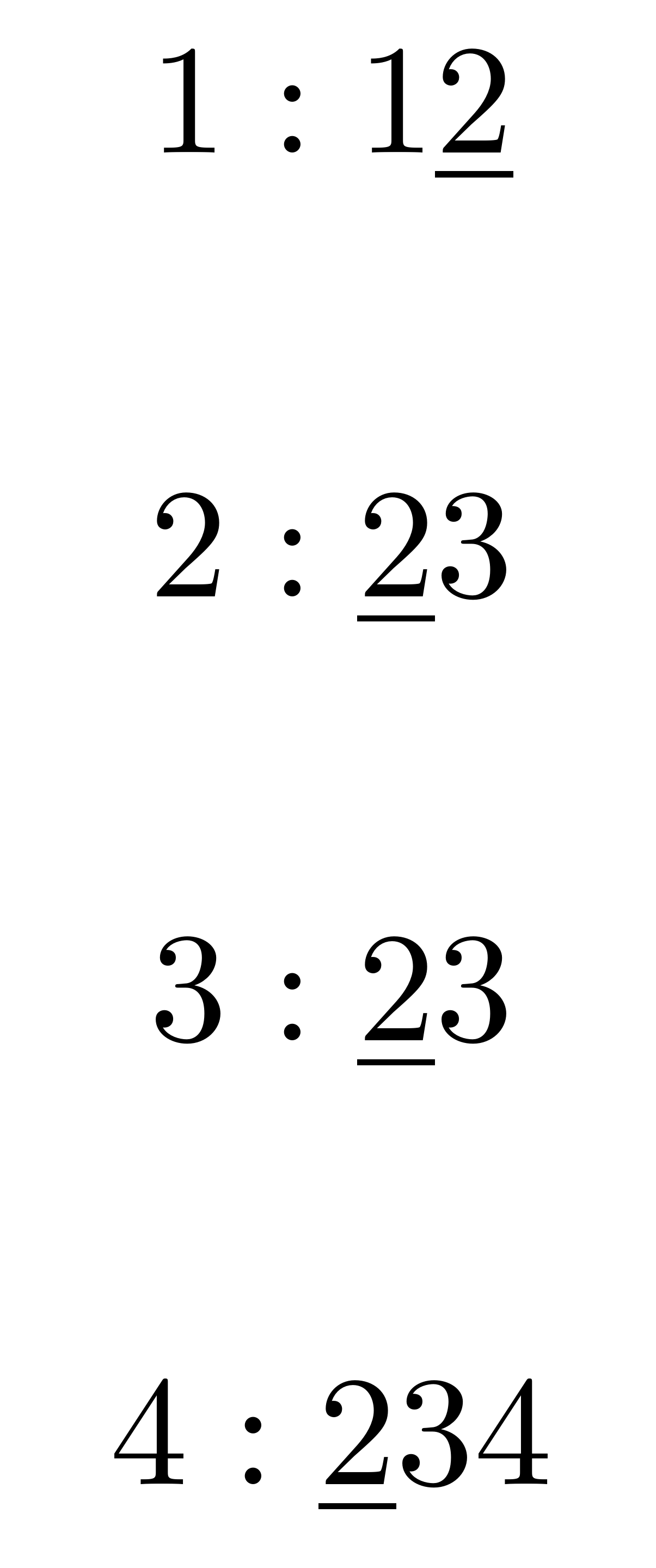}&\includegraphics[width=0.1\linewidth]{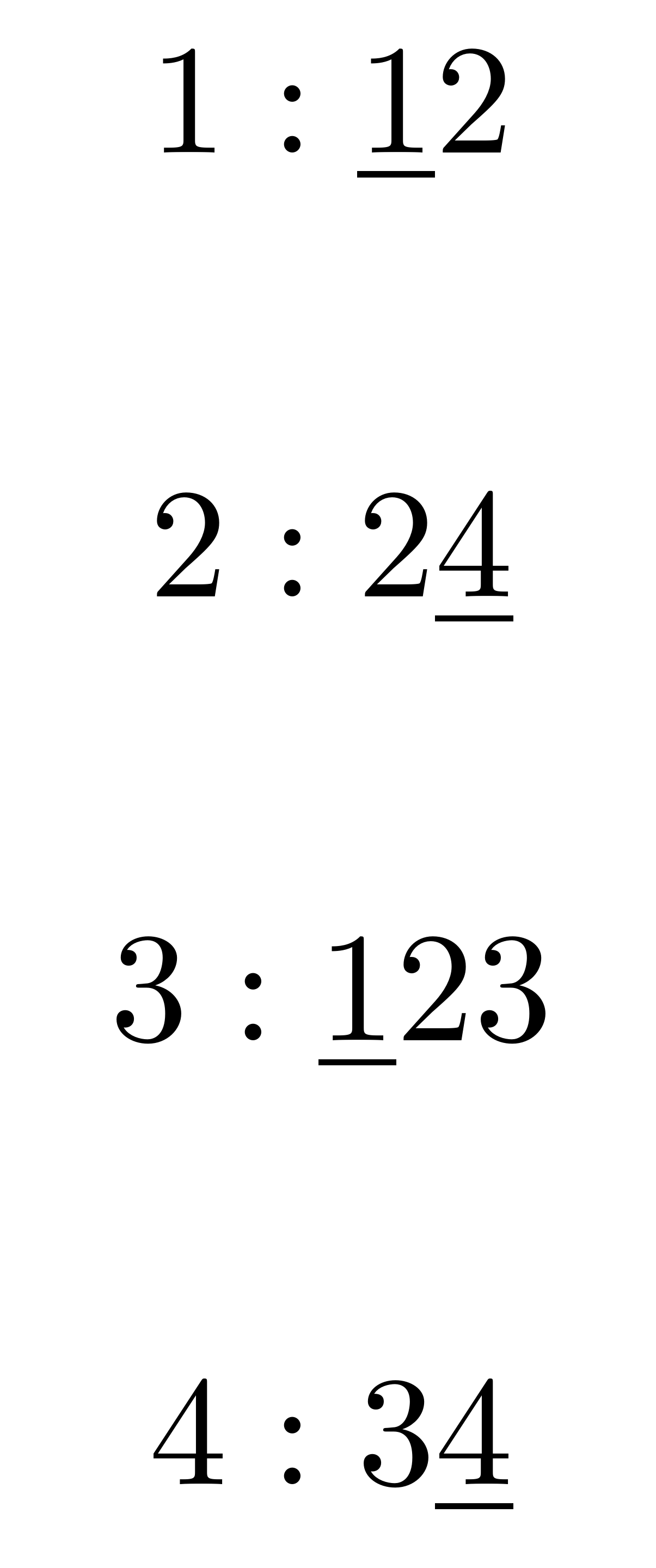}&\includegraphics[width=0.1\linewidth]{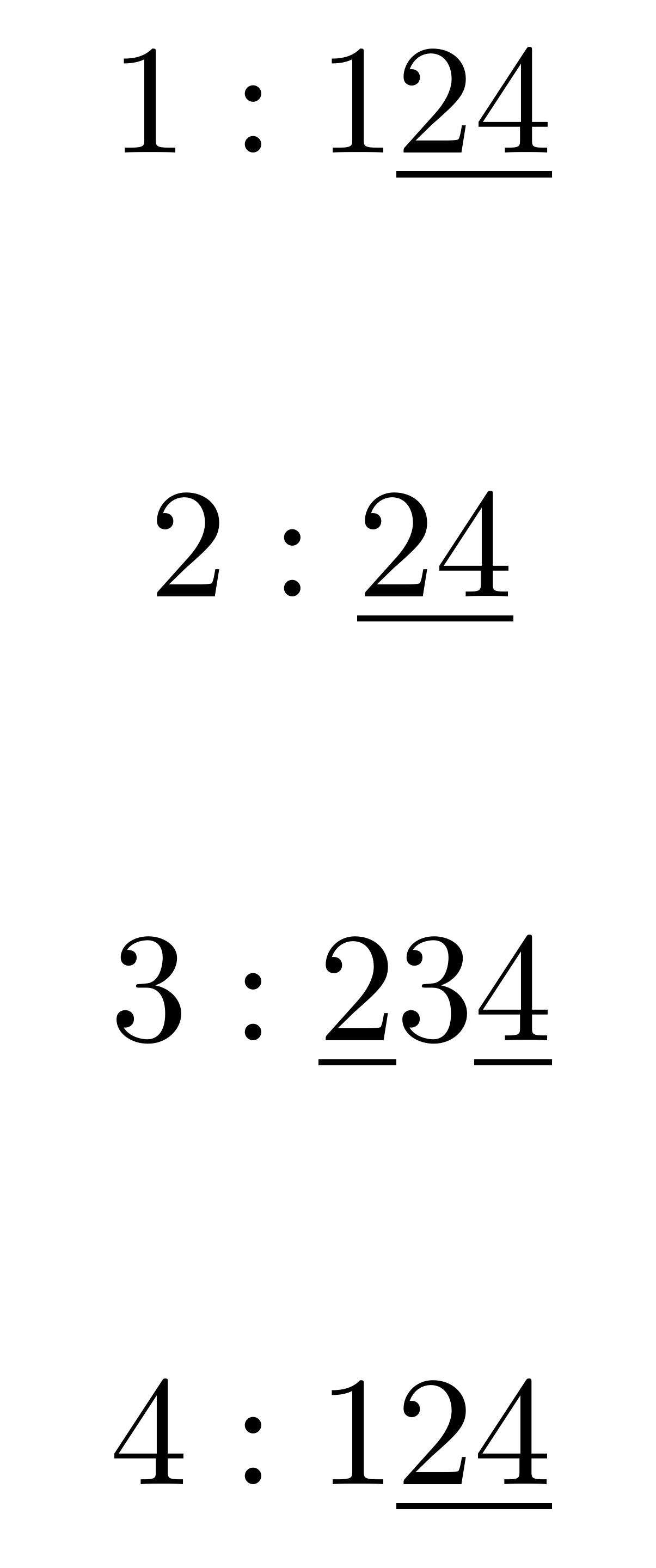}\\
   & Broadcast & Cover of size $k$& $k$-broadcast of $k$ processes \\
   & Example for $n=4$& Example for $n=4$ and $k=2$& Example for $n=4$ and $k=2$ \\
   [0.5ex]
   \hline
    \Tstrut \Bstrut Lower Bound:& \cite{schwarz2017linear}: $\left\lceil{\frac{3n-1}{2}}\right\rceil-2$  &
   $\left\lceil{\frac{3(n-k)}{2}}\right\rceil-1$&$\left\lceil{\frac{3(n-3k)}{2}}\right\rceil+2$\\
   [1ex]
   \hline
   \rule{0pt}{3ex} \rule[-1.6ex]{0pt}{0pt}
   & \cite{fugger2020radius}: $O(n\log\log n)$&&\\
   [1ex]
   \rule{0pt}{3ex} \rule[-1.6ex]{0pt}{0pt} Upper Bound:& $\ceil { (1+\sqrt{2})n}$&$\frac{\pi^2+6}{6}n+1$&$\ceil { (1+\sqrt{2})n}+k-1$\\
   \hline
\end{tabular}}
\end{center}

    \caption{Summary of models and results. The adversary chooses a sequence of graphs as stated in order to delay the objective as much as possible. In the objective examples, $x:y$ means that $y$ is an in-neighbour of $x$ at the round the objective is attained. Underlines are for emphasis purposes only.}
    \label{table}
\end{figure}

\paragraph*{Organization}

The remainder of this paper is organized as follows.
Section~\ref{tools} introduces some basic tools that will be useful throughout the paper, and presents first insights. In Section~\ref{tree}, we give the linear upper bound for broadcast time in our model. Section~\ref{multicover} and Section~\ref{multiroot} respectively showcase our results for the cover on $k$-forests and the $k$-broadcast on $k$-rooted networks. After reviewing related work in Section~\ref{sec:relwork}, we conclude and provide some future research directions in Section~\ref{sec:conclusion}. All the lower bound results can be found in appendix~\ref{appendix}, as well as other basic proofs.

\section{Basic Tools}\label{tools}

In this section, we define generalizations of the in- and out-neighborhoods in the product graph, and some basic properties these tools follow. The proofs of these properties being basic, we defer them to the appendix.

\begin{definition}
    The in-neighborhood of process $x$ between rounds $t$ and $t'$, $t,t'\geq 0$, denoted by $\I{t} {t'} x$, is defined as follows: If $t \leq t', \I{t} {t'} x$ is the in-neighborhood of process $x$ in the graph $G_t\circ\hdots\circ G_{t'}$. If $t=t'+1$, set $\I{t} {t'} x=\{x\}$. Otherwise, $\I{t} {t'} x=\varnothing$.
\end{definition}

\begin{definition}
    The out-neighborhood of process $x$ between rounds $t$ and $t'$, $t,t'\geq 0$, denoted by $\Out{t} {t'} x$, is defined as follows: If $t \leq t', \Out{t} {t'} x$ is the out-neighborhood of process $x$ in the graph $G_t\circ\hdots\circ G_{t'}$. If $t=t'+1$, set $\Out{t} {t'} x=\{x\}$. Otherwise, $\Out{t} {t'} x=\varnothing$.
\end{definition}

We prove in Appendix~\ref{appendix} that these generalized neighborhoods behave as expected:

\begin{restatable}{lemma}{first}\label{easy}
    Let $t,t' \geq 0, x,y \in [n]$. Then $x \in \Out {t} {t'} y \Leftrightarrow y \in \I {t} {t'} x$.
\end{restatable}

\begin{restatable}[Transitivity]{lemma}{outout}\label{transitive}Let $t, t', t''\geq 0$, and $x,y,z \in [n]$. We have the following properties:

    \label{outout}\emph{i.} If $y \in \Out t {t'} x $ and $z \in \Out {t'+1} {t''} y$, then $z \in \Out t {t''} x$.
    
    \label{inout}\emph{ii.} If $x \in \I t {t'} y $ and $z \in \Out {t'+1} {t''} y$,then $z \in \Out t {t''} x$.
    
    \label{inin}\emph{iii.} If $x \in \I t {t'} y $ and $y \in \I {t'+1} {t''} z$, then $z \in \Out t {t''} x$.
   \end{restatable}

And finally, we show that these neighborhoods only grow over time:

\begin{restatable}[Monotonicity]{lemma}{trivialinclusion}\label{trivialinclusion}
    If in each round, all nodes have a self-loop, then for any $t_1\leq t_2$ and $ t_3 \leq t_4$, for any process $x$ we have:
    
    \emph{i.} $\I {t_2} {t_3} x \subseteq \I {t_1} {t_4} x$.
        
    \emph{ii.} $\Out {t_2} {t_3} x \subseteq \Out {t_1} {t_4} x$.
\end{restatable}

\section{Broadcasting on Trees}\label{tree}

In this section, we focus on the fundamental problem of broadcasting on dynamic trees. We give an upper bound for the problem, before recalling a lower bound.

\subsection{The Upper Bound}

We will show that the key to understand how information propagate is to consider what the root knows -- or the in-neighbors of the root -- before the beginning of every round. We will show that the root must either have a lot of in-neighbors that were roots in previous rounds, or many in-neighbors in general. We will then show that any in-neighbor of the root before a round has at least one more out-neighbor after the round than before it. We will finally show that any adversary that tries to balance these two facts will fail to prevent broadcast for a time longer than linear.

\begin{definition}
    Let $t$ be a round. We denote by $r_t$ the root of $G_t$, and call it the root of the round $t$.
\end{definition}

Harnessing the fact that there always exists a path from a root to any other process in a network, we give the two following lemmas:

\begin{lemma}\label{propagation}
 Let $t, t'$ be rounds such that $t\leq t'$. We have that:
 \begin{enumerate}[i]
     \item \label{ingetsbig}If $x$ is a process such that $r_t \notin \I {t+1} {t'} x$, then $\card {\I t {t'} x}> \card {\I {t+1} {t'} x}$.
     \item \label{outgetsbig}If $x$ is a process such that $r_{t'} \in \Out {t} {t'-1} x$, then $\card {\Out t {t'} x}> \card {\Out {t} {t'-1} x}$, unless $\card {\Out {t} {t'-1} x}=n$.
 \end{enumerate}
\end{lemma}

\begin{proof}
    \ref{ingetsbig}. We will show that there exists a process $y \in \I {t+1} {t'} x$ that has an in-neighbor $z$ in $G_t$ such that $z \notin \I {t+1} {t'} x$. Then, by Transitivity (Lemma~\ref{transitive}), $z \in \I t {t'} x$. By Monotonicity (Lemma~\ref{trivialinclusion}), we have $\I {t+1} {t'} x \subset \I {t} {t'} x$, this will show that $\card {\I t {t'} x}> \card {\I {t+1} {t'} x}$.

    \noindent Let us now find such a $z$. Consider the path from $r_t$ to $x$ in $G_t$. Since $r_t \notin \I {t+1} {t'} x$, and trivially $x \in \I {t+1} {t'} x$, this path must include an edge $(z,y)$ such that $z \notin \I {t+1} {t'} x$, $y \in \I {t+1} {t'} x$.

    \ref{outgetsbig}. Let us look at the case $\card {\Out {t} {t'-1} x}<n$. We will show that there exists a process $y \in \Out {t} {t'-1} x$ that has an out-neighbor $z$ in $G_{t'}$ such that $z \notin \Out {t} {t'-1} x$. Then, by Transitivity (Lemma~\ref{outout}), $z \in \Out t {t'} x$. By Monotonicity (Lemma~\ref{trivialinclusion}), we have $\Out {t} {t'-1} x \subset \Out {t} {t'} x$, this will show that $\card {\Out t {t'} x}> \card {\Out {t} {t'-1} x}$.

    \noindent Let us now find such a $z$. Since $\card {\Out {t} {t'-1} x}<n$, there exists a process $a$ such that $ a \notin \Out {t} {t'-1} x$. Consider the path from $r_{t'}$ to $a$ in $G_{t'}$. Since $r_{t'} \in \Out {t} {t'-1} x$, and $a \notin \Out{t} {t'-1} x$, this path must include an edge $(y,z)$ such that $z \notin \Out {t} {t'-1} x$, $y \in \Out {t} {t'-1} x$.
\end{proof}

The following lemma will link the number of in-neighbors a node has to the number of in-neighbors it has among the roots of the preceding rounds:

\begin{lemma}\label{manyones}
    Let $x$ be a process, and $t_1, t_2$ be rounds such that $t_1 \leq t_2$. Then:
    $$
    \card{\{t: t_1\leq t\leq t_2, r_t \notin \I {t_1} {t_2} x\}}+1 \leq \card{\I {t_1} {t_2} x}
    $$
\end{lemma}

This will be proven using Lemma~\ref{propagation}.\ref{ingetsbig}, since every time $r_t \notin \I {t_1} {t_2} x$, $\I {t_1} {t_2} x$ gets larger:

\begin{proof}
    Let $A=\{t: t_1\leq t\leq t_2, r_t \notin \I {t_1} {t_2} x\}$. Then, in particular, for any $ t \in A$, we have $r_t \notin \I {t+1} {t_2} x$. Then, for all $t \in a$, applying Lemma~\ref{propagation}.\ref{ingetsbig}, we have $\card{\I {t} {t_2} x} > \card{ \I {t+1} {t_2} x }$. Let $A=\{t^1, \hdots, t^k\}$, with $t^i<t^{i+1}$ for any $1\leq i<k$. Then:
    $$
    \card{\I {t_1} {t_2} x} \geq \card{\I {t^1} {t_2} x } >\card{\I {t^1+1} {t_2} x } \geq \card{\I {t^2} {t_2} x }>\card{\I {t^2+1} {t_2} x }\geq \hdots >\card{\I {t^k+1} {t_2} x }\geq 1
    $$

    Where the non-strict inequalities derive by Monotonicity(Lemma~\ref{trivialinclusion}), and the last one from the fact that $t^k+1 \leq t_2+1 \Rightarrow x \in \I {t^k+1} {t_2} x$. There are $k=\card{A}$ strict inequalities over integers, which concludes the proof.
\end{proof}

We now define the \emph{rounds graph}, which will keep track of the information -- the in-neighbors -- the root of each round has:

\begin{definition}
    We define the \emph{rounds graph} as follows:

    The graph has $2n+\ceil{\sqrt2 n}$ nodes: one node representing each process, and one node for each of the first $\ceil{(1+\sqrt{2})n}$ rounds.

    And it has the directed edges: one edge from a process $p$ to a round $t$ if $p \in \I 1 {t-1} {r_t}$, and one edge from a round $t<\ceil{\sqrt 2 n}$ to a round $t'>t$ if $r_t \in \I 1 {t'-1} {r_{t'}}$.
\end{definition}

We will now show that there is at least a node of out-degree $n$ in that graph, which will translate into a process that has broadcast its piece of information to everyone:

\begin{lemma}\label{lem:pig1}
    In the rounds graph, there is at least a node of out-degree $n$.
\end{lemma}

\begin{proof}
    Let us look at round $t$. If $t\leq\ceil{\sqrt 2 n}$, then $t$ has in-degree at least $t$. Indeed, it has in-degree:
    \begin{multline*}
    \card{\{t': 1\leq t'\leq t-1, {r_{t'} }\in  \I {1} {t-1} {r_t}\}}+ \card{\I {1} {t-1} x}\\ \geq 1+\card{\{t': 1\leq t'\leq t-1, {r_{t'} }\in \I {1} {t-1} {r_t}\}}+\card{\{t': 1\leq t'\leq t-1, {r_{t'} }\notin \I {1} {t-1} {r_t}\}}=t
    \end{multline*}
    where we used Lemma~\ref{manyones} for the inequality.
    Similarly, if $t>\ceil{\sqrt 2 n}$, then $t$ has in-degree at least $\ceil{\sqrt 2 n}$. Indeed, it has in-degree:
    \begin{multline*}
    \card{\{t': 1\leq t'< \ceil{\sqrt 2 n}, {r_{t'} }\in  \I {1} {t-1} {r_t}\}}+ \card{\I {1} {t-1} {r_t}}\\ \geq 1+\card{\{t': 1\leq t'< \ceil{\sqrt 2 n}, {r_{t'} }\in \I {1} {t-1} {r_t}\}}+\card{\{t': 1\leq t'\leq t-1, {r_{t'} }\notin \I {1} {t-1} {r_t}\}}\\
    \geq 1+\card{\{t': 1\leq t'< \ceil{\sqrt 2 n}, {r_{t'} }\in \I {1} {t-1} {r_t}\}}+\card{\{t': 1\leq t'<\ceil{\sqrt 2 n}, {r_{t'} }\notin \I {1} {t-1} {r_t}\}}\\=\ceil{\sqrt 2 n}
    \end{multline*}
    where we used Lemma~\ref{manyones} for the first inequality.

    Summing the in-degrees over all the rounds, we get that the number of edges $\card{E}$ is at least:

    $$
    \card{E}\geq \sum_{t=1}^{\ceil{\sqrt 2 n}}t+n\times \ceil{\sqrt 2 n} = \frac {\ceil{\sqrt 2 n}\left(\ceil{\sqrt 2 n}+1\right)} 2 + n \ceil{\sqrt 2 n}> \ceil{(1+\sqrt 2) n}n
    $$

    But only the $n$ nodes representing the processes and the nodes representing the first $\ceil{\sqrt 2 n}-1$ rounds have out edges. The pigeonhole principle asserts then that one of those nodes has an out-degree of at least $n$.
\end{proof}

\begin{theorem}
    $\T(\RT_n) \leq \ceil{(1+\sqrt 2)n}$
\end{theorem}

\begin{proof}
    Assume it is not the case, that is, for every process $x \in [n]$, for every round $t\leq \ceil{(1+\sqrt 2)n}$, $\card{\Out 1 t x}<n$. We know that in the rounds graph, there is a node $y$ of degree at least $n$. Define $z$ as follows: if $y$ represents a process, let $z$ be that process. If $y$ represents a round, let $z$ be the root of that round. We will show that $z$ must have broadcast before $\ceil{(1+\sqrt 2)n}$ rounds.

    Let $t_1< \hdots < t_n$ be the rounds $y$ has out-edges to. By definition, this means that $z\in \I 1 {t_i-1} {r_{t_i}} \Leftrightarrow r_{t_i} \in \Out 1 {t_i-1} z$ for every $i \in [n]$. By Lemma~\ref{propagation}\ref{outgetsbig}, we thus have, for every $i \in [n]$, that $\card{\Out 1 {t_i} z} >\card{ \Out 1 {t_i-1} z}$. Then, using  Monotonicity (Lemma~\ref{trivialinclusion}) for non-strict inequalities:
    $$
    \card{\Out 1 {t_n} z} >\card{ \Out 1 {t_n-1} z}\geq \card{\Out 1 {t_{n-1}} z} >\card{ \Out 1 {t_{n-1}-1} z}\geq \hdots \geq\card{\Out 1 {t_1} z} >\card{ \Out 1 {t_1-1} z}\geq 0
    $$

    We have $n$ strict inequalities over non-negative integers, the largest one must be at least $n$, which is a contradiction.    
\end{proof}

\subsection{The Lower Bound}
A lower bound for this problem has been given by Zeiner, Schwarz, and Schmid~\cite{schwarz2017linear}:

\begin{restatable}{theorem}{lowerb}
   $\T(\RT_n) \geq \ceil{\frac{3n-1}{2}-2}$
\end{restatable}

A figure of that lower bound can be found in Appendix~\ref{appendix}.

\section{\texorpdfstring{Covering on $k$-Forests}{Covering on k-Forests}}\label{multicover}

In this section, we study an adversary that has to choose a communication network in each round that is a union of $k$ trees. In this setting, we cannot ensure broadcast, so we look at the time when there exists a cover of size $k$: $k$ processes such that any other process has heard of at least one of them. 

\subsection{The Upper Bound}

Even though the problem is very similar to broadcasting on trees, the proofs of Section~\ref{tree} do not translate in a straightforward way into an upper bound for covering on $k$-forests. We thus have a completely different proof for this problem.

The intuition of our approach is as follows: We will start with a cover of size $n$ at some time $t'=\expo t n$ that is large enough, and then go back in time until we find a process that can reach two other processes, say, $x$ and $y$, of that cover. Calling this process $p_{n-1}$ and the corresponding time $\expo t {n-1}$, we thus have $p_{n-1} \in \I {\expo t {n-1}} {\expo t n} x$ and $p_{n-1} \in \I {\expo t {n-1}} {\expo t n} y$. When repeating the process, we then remove $x$ and $y$ from our set of processes to cover, add $p_{n-1}$, and start over, until the cover has size $k$. 
We need to be careful to guarantee that rounds do not overlap.

Indeed, to remove $p_{n-1}$ from our set of processes to cover, we have to reach it before round $\expo t {n-1}$, otherwise we will not be guaranteed to reach $x$ or $y$ at time $\expo t n$. Thus, we will store with each process $x$ of the cover the corresponding round $t$ such that $x$ has to be reached by round $t$ by the process that replaces $x$ in the cover.
More specifically, we model the cover by a series of sets $(A_u)_{k\leq u\leq n}$, where each $A_u$ is a collection of $u$ pairs $(p,t)$, where $p$ is a process, and $t$ is a round. To compute $A_{u-1}$ from $A_u$, we have to find a process that can reach two processes $p_1$, $p_2$ by rounds $t_1$, $t_2$, such that $(p_1, t_1), (p_2, t_2) \in A_u$ and then we replace these two pairs by a new pair, creating the cover $A_{u-1}$.

In this section we first state the definitions and results of this section, before giving the full proof to each of our claims.
We first define what it means for a set $A$ of pairs $(p,t)$ to be a cover.

\begin{definition}
    A set $A=\{(a_1, t_1), \hdots , (a_s,t_s)\}$ of $s$ (process, round) pairs is a {\em cover of a set $B$ of processes for round $t\geq\max_s \{t_s\}$} if for every $b\in B$, there exists an $i \in [s]$ such that $b \in \Out {t_i+1} {t} {a_i}$.
\end{definition}

We next couple the cover property of set $A$ with \emph{strictness}, which 
indicates 
that we did not (yet) go back enough in time to find a process that reaches two different processes in $A$.

\begin{definition}
    A set $A=\{(a_1, t_1), \hdots , (a_s,t_s)\}$ of $s$ (process, round) pairs is {\em strict at round $t$} if there exists no process $p\in [n]$ and $i,j \in [s], i \neq j$ such that $p \in \I t {t_i} {a_i}$ and $p \in \I t {t_j} {a_j}$.
\end{definition}

As we consider earlier and earlier rounds, the sets $\I t {t_i} {a_i}$ will get larger and larger, and $A$ will lose its strictness. Thus, we then define the following sequence of covers of $[n]$, and analyze their strictness over time. We carefully choose our set $A_s$ so that it has cardinality $s$. This means that $A_k$ has cardinality $k$, which is our goal.

\begin{restatable}{definition}{setsdefinition}\label{thesets}
    Let $t'$ be a large enough round. For every $k\leq s \leq n$, we define a sequence of strict sets $A_s$ and rounds $\expo t s$ as follows:

    Define $A_n=\{(i, t'): i \in [n]\}$, $\expo t n=t'$.

    Define $\expo t s=\max_{j\in\N}\{A_{s+1} \text{ is not strict at round } j\}$. 
    As $A_{s+1}$ is not strict at round $\expo t s$, there exist $i,j\in [s+1]$ and a process $p_s$ such that $p_s \in \I {\expo t s} {t_i} {a_i} \inter \I {\expo t s} {t_j} {a_j}$. If $(p_s, \expo t s-1) \in A_{s+1}$, 
    we define $A_{s}=A_{s+1}\setminus\{(a_i, t_i)\}$, else
    we define 
     $A_{s}=(A_{s+1}\setminus\{(a_i, t_i), (a_j,t_j)\})\union \{(p_s, \expo t s-1)\}$.

     Define  $\Delta_s=\expo t s-\expo t {s-1}$  for $k+1 \le s \le n$.
\end{restatable}

\noindent 
Recall that our goal is to upper bound $t'$, which can be done if
we upper bound $\sum_{s = k+1}^n \Delta_s$.
The strictness of a set is a key notion as a strict set has the following very useful property:

\begin{restatable}[Strict Increments]{lemma}{strictincrements}\label{bigin}
    Let $s \ge k + 1$ and let $A=\{(a_1, t_1), \hdots , (a_s, t_s)\}$ be a strict set of size $s$ at round $t$. Let $I\subseteq \{ i \in [s]: t \leq t_i+1\}$. Then there exists a set of indices $J\subseteq I$, $\card J \geq \card I-k$, such that for every $i \in J$, $\card {\I {t-1} {t_i} {a_i}} > \card {\I {t} {t_i} {a_i}}$.
\end{restatable}

\begin{proof}  Consider a root $r$ of round $t-1$. As $A$ is strict at round $t$ is follows that 
    $r \in \I {t} {t_i} {a_i}$ for  at most one $a_i$ with $i \in I$. As there are at most
    $k$ roots in round $t-1$, it follows that there are at least $\card I-k$ values $i \in  I $ such that none of the roots of round $t-1$ is in $\I {t} {t_i} {a_i}$. Let $J$ be the set of all such values of $i$. Let us denote for any $p \in [n]$  the root of the tree that contains $p$ in round $t$ by the value $r_t(p)$. It follows that in particular
    $r_{t-1}(a_i) \notin \I {t} {t_i} {a_i}$. Since $t\leq t_i+1$, we have that $a_i \in \I {t} {t_i} {a_i}$.
    Now for each such $i$ the path from $r_{t-1}(a_i)$ to $a_i$ in $G_{t-1}$ must contain an edge $(x,y)$ such that $x \notin \I {t} {t_i} {a_i}$, and $y \in \I {t} {t_i} {a_i}$. By Transitivity (Lemma~\ref{transitive}), it holds that $(\I{t} {t_i} {a_i}\union\{x\}) \subseteq \I{t-1} {t_i} {a_i}$, which implies that $\card {\I{t-1} {t_i} {a_i}} > \card {\I{t} {t_i} {a_i}}$.
\end{proof}

\noindent It is not hard to show that \emph{all}
$a_i$ of $A_s$ fulfill $\expo t s \le t_i+1$. The following lemma helps us find more sets that satisfy the conditions of the Strict Increments Lemma. This essentially follows from the fact that, by construction,  at least $2s-u$ elements from $A_u$ are shared with $A_s$.

\begin{restatable}{lemma}{largesetoft}\label{largesetoft}
    Let $s,u\in \{k, \hdots , n\}$ such that $s\leq u$. Let $A_s=\{(a_1, t_1), \hdots, (a_s, t_s)\}$. Let $I=\{i \in [s]: \expo t u \leq t_i+1\}$. Then it holds that $\card I \geq 2s-u$.
\end{restatable}

\noindent We now define the strict rounds graph, which we use to analyze the values of $\Delta_s$ as $s$ varies from $n$ to $k$. A depiction of that graph can be seen in \figref{fig:strict}.

\begin{restatable}{definition}{strictrounds}
    The \emph{strict rounds graph} $(V,E)$ consists of $n-k$ vertices, labeled from $k+1$ to $n$, where each vertex $s$ has weight $\Delta_s$.
    There exists a directed edge from vertex $u$ to vertex $s$  if $s\leq u \le \min\{2s-k-1,n\}$, and its weight is $w(u,s)=2s-k-u$.
\end{restatable}

\begin{figure}
    \centering
    \includegraphics[width=\linewidth]{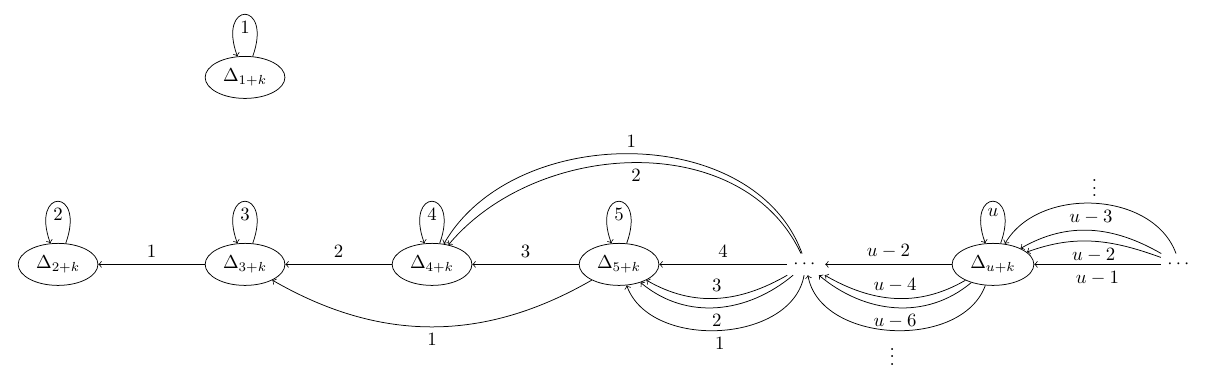}
    \caption{The strict rounds graph. Vertex $u$'s weight represents $\Delta_{u}=\expo t u-\expo t {u-1}$. The weight of an edge $(u,s)$ represents how much a round $t \in [\expo t {u-1}+1, \expo t u]$ contributes at least to $S_s$. Because of the strictness of $A_{s}$, the sum of the weights of the in-neighbors (multiplied by the edges' weight) of a node $s$, which is smaller than $S_s$, should not exceed $n$.}
    \label{fig:strict}
\end{figure}

\noindent The next lemma is the crucial lemma: To bound $t'$ we first bound the following ``weighted volumes''. If we define $\alpha_s$ to be the weighted out-degree of a node $s$, and the \emph{volume} of $s$ to be $\alpha_s$  multiplied by its own weight, then in the strict rounds graph, the cumulative sum over the volumes of the first $j$ vertices is at most $j \cdot n$. 

\begin{restatable}{lemma}{corsum}\label{cor:sum}
  Let $u\in \{k+1, \hdots , n\}$.
  Then
  $\sum_{s \le u} \Delta_s \alpha_s \le (u-k) \cdot n$.
\end{restatable}

We briefly sketch the proof of this lemma. We first prove that  for every $k+1 \le w \le u$ the sum
$\sigma_w := \sum_{v=w}^{\min\{2w-k-1, n\}}\left(2w-k - v\right)\Delta_{v}$ is at most  $n$ as follows.
Since for every $k+1 \le w \le u$ the set $A_w$ is strict at all rounds 
$t \ge \expo t {w-1}+1$, it follows that $S_w := \sum_{(a_i, t_i) \in A_w} \card {\I {\expo t {w-1}+1} {t_i} {a_i}}$ is at most $n$. 
We then lower bound $S_w$ by  $\sum_{v=w}^{\min\{2w-k-1, n\}}\left(2w-k - v\right)\Delta_{v}$ in two steps: First, Lemma~\ref{largesetoft} allows us to find a set $I$ of cardinality larger than $2w-v$ for each round in the interval $[\expo t {v-1} +1, \expo t v]$ that fits the Strict Increments Lemma's conditions. We then use the Strict Increments Lemma to show that $S_w$ increases by $2w-v-k$ in each of those rounds, which results in a lower bound of $\sigma_w$ for $S_w$, and thus the upper bound $\sigma_w\le n$. Summing this inequality over all $u-k$ nodes $w$ with $w \le u$ gives $\sum_{w\leq u}  \sigma_w \leq (u-k)\cdot n$.
 
Next note that $\sigma_w$  is the weighted in-degree of node $w$ in the strict rounds graph where each edge is weighted by the product of its edge weight and the weight of its tail.
By definition, the tail of every (directed) edge has a higher label than its head and, thus, every outgoing edge of a node $s \le u$ is also an incoming edge of a node $w \le u$. This allows us to argue that
$\sum_{w \le u} \sigma_w$ is at least $\sum_{s \le u} \Delta_s \alpha_s$, which leads to the final result.

We use Lemma~\ref{cor:sum} to bound $t'$ as follows. We first show that any vertex weight distribution on the strict rounds graph following the volume bound must fulfill the following property. 

\begin{restatable}{lemma}{littledeltas}\label{littledeltas}
    Let $\{\delta_s\}_{k+1\leq s\leq n}$ with $\delta_s \in \mathbb{R}$ be a vertex weight distribution  over the strict rounds graph such that for every $u\in \{k+1, \hdots n\}$, $    \sum_{s \leq u} \delta_s \alpha_s \leq (u-k)\cdot n   $. Then $\sum_{s=k+1}^{n}\delta_s\leq \frac{\pi^2+6}{6} n$. 
\end{restatable}
\noindent Lemma~\ref{cor:sum} and Lemma~\ref{littledeltas} give the desired bound on  $\sum_{s=k+1}^{n}\Delta_s$, which in turn bounds $t'$.
\begin{restatable}{corollary}{atmostbeta} \label{corollary}
    $\sum_{s=k+1}^{n}\Delta_s\leq \frac{\pi^2+6}{6} n$
\end{restatable}

\begin{restatable}{theorem}{bigtheorem}
    $t^{c}_k(\F_n^k) \leq \frac{\pi^2+6}{6}n+1$
\end{restatable}

The rest of this section is dedicated to proving in detail all of those claims, as well as introducing all concepts, stating and proving any intermediate lemmata that would be necessary.
First we analyze in detail the sets defined in Definition~\ref{thesets}, making sure they have the desired cardinality and proving that they form a cover of $[n]$ from round $1$ to $t'$.
Note that it directly follows from Definition~\ref{thesets} that the size of $A_s$ is $s$ and we will in the following
always use $(a_i, t_i)$ for $1 \le i \le s$ to denote the elements of the set $A_s$.

\begin{lemma}
    For every $s \in \{k,\hdots, n\}, \card{A_s}=s$. 
\end{lemma}

\begin{proof}
    It is true by reverse induction. Indeed, it is trivially true for $s=n$. Suppose it is true for some $s \in \{k+1, \hdots, n\}$ and consider two cases:
    If $(p_{s-1}, \expo t {s-1}-1) \in A_{s}$, then $\card{A_s}-\card{A_{s-1}} = 1$, as only $(a_i,t_i)$ is removed from $A_s$.
    If $(p_{s-1}, \expo t {s-1}-1) \not\in A_{s}$, then both $(a_i,t_i)$  and $(a_j, t_j)$ are removed from $A_s$ and $(p_{s-1}, \expo t {s-1}-1)$ is added, which implies that $\card{A_s}-\card{A_{s-1}} = 1$.
    Thus, in both cases $\card{A_s}-\card{A_{s-1}} = 1$. As by induction $|A_s| = s$, it follows that $|A_{s-1}| = s -1$.
\end{proof}

In Definition~\ref{thesets}, we offset $\expo t s$ by one when introducing $p_s$ in $A_s$, which guarantees that $A_s$ is a cover of $[n]$ for round $t'$, as shown below. 

\begin{lemma}\label{cover}
    For every $s \in \{k, \hdots, n\}$, we have that $A_s$ is a cover of $[n]$ for round $t'$.
\end{lemma}

\begin{proof}
    We show the claim by reverse induction. It holds trivially for $s=n$.
    Assume now that $A_{s+1}$ is a cover of $[n]$ for round $t'$. We will show that this implies that $A_{s}$ is also a cover of $[n]$ for round $t'$. By the fact that $A_{s+1}$ is a cover, it follows for each $x\in [n]$ that there exists a $h \in [s+1]$ such that $x\in \Out {t_h+1} {t'} {a_h}$. If $(a_h, t_{h}) \in A_{s}$, then there is nothing to do. If however $(a_h, t_{h}) \notin A_{s}$, then $p_s\in \I {\expo t {s}} {t_h} {a_h}$, and by Transitivity (Lemma~\ref{inout}) it follows that $x \in \Out {\expo t {s}} {t'} {p_s}$.
\end{proof}

Now that we are assured that the cover property will hold throughout, we give the intuition of what follows. As defined above, we have $\Delta_u=\expo t u-\expo t {u-1}$ for $k < u \le n$. We will then look at $A_s$ for some $s \in \{k, \hdots , n\}$. Since $A_s$ is strict at round $\expo t {s-1}+1$, it follows that $\left(\I {\expo t {s-1}+1} {t_i} {a_i}\right)_{i\in[s]}$ are pairwise disjoint subsets of $[n]$, which implies that
$S_s = \sum_{(a_i, t_i) \in A_s} \card {\I {\expo t {s-1}+1} {t_i} {a_i}}$ is at most $n$. We will find a lower bound for that value that depends on the values of $\Delta_u$.

To do so, we will use the Strict Increments Lemma (Lemma~\ref{bigin}). Indeed we will first prove that for every $i \in [s]$, it holds that $\expo t s\le t_i+1$. This will allow us to use the Strict Increments Lemma (Lemma~\ref{bigin}) between rounds $\expo t {s-1} +2$ and $\expo t s$ 
and so each of those rounds contributes an additive $s-k$ to the lower bound of $S_s$. Of course, such a contribution only happens if $\expo t {s-1}+2 \le \expo t s$.

We will further prove that there exist $s-1$ elements $i\in [s]$ such that $\expo t {s+1}\le t_i+1$. This will follow from the fact that $\card {A_s \inter A_{s+1}} \geq s-1$, and that for every $(a,t) \in A_{s+1}$, we have that $\expo t {s+1} \leq t+1$. This allows us to use the Strict Increments Lemma (Lemma~\ref{bigin}) between rounds $\expo t {s} +1$ and $\expo t {s+1}$, so that each of those rounds contributes at least an additive $s-1-k$ to the lower bounds of $S_s$ \emph{in addition to the contribution of rounds $[\expo t {s-1}+2, \expo t s]$.}

In fact, we will generalize this analysis to all values of $u \ge s$ with $A_u \inter A_s \neq \varnothing$, and for each $u$ every round in
 $[\expo t {u-1}+1, \expo t u]$ contributes at least $2s -k - u$ to the lower bound of $S_s$, leading to a contribution at at least $(2s-k-u) \Delta_u$ for $u$.

\begin{lemma}
    $A_s$ is strict at round $\expo t s+1$, for any $s \in \{k, \hdots, n\}$.
\end{lemma}

\begin{proof}
    Let $s \in \{k, \hdots, n\}$. $A_{s+1}$ is strict at round $\expo t s+1$. Assume by contradiction that $A_s$ is not strict at round $\expo t s +1$, that is, there exists a $x \in [n]$ and $i,j \in [s]$ such that $x \in \I {\expo t {s}+1} {t_i} {a_i}$ and $x \in \I {\expo t {s}+1} {t_j} {a_j}$. Since $\I {\expo t {s}+1} {\expo t s-1} {p_s} = \varnothing$ because $\expo t {s}+1 > (\expo t s-1)+1$, we have that $x \notin \I {\expo t {s}+1} {\expo t s-1} {p_s}$. This implies that $(a_i, t_i), (a_j, t_j) \in A_{s+1}$, which contradicts the strictness of $A_{s+1}$. This concludes the proof.
\end{proof}

\begin{corollary}\label{rightorder}
    $\expo t s \leq \expo t {s+1} \quad \forall s \in \{k,\hdots , n-1\}$ 
\end{corollary}

\begin{corollary}\label{nicelyordered}
    For every $s \in \{k, \hdots, n\}$, we have that for every $i \in [s]$, $\expo t {s} \leq t_i+1$.
\end{corollary}

\begin{proof}
    By reverse induction, it is true for $s=n$, as for every $i \in [n], t_i=\expo t n$. Assume it is true for some $s+1$ for $s \in \{k, \hdots n-1\}$, and let us look at $A_s$. Then by Corollary~\ref{rightorder} and the induction hypothesis, for every $i \in [s]$ we have that either $(a_i, t_i) \in A_{s+1}$ and thus $\expo t {s-1} \leq \expo t s\leq t_i+1$, or that $(a_i, t_i) = (p_s, \expo t s-1)$ and trivially the inequality holds.
\end{proof}

\begin{lemma}\label{biginter}
    Let $s,u \in \{k, \hdots, n\}$ such that $s \leq u$. Then $A_s \inter A_u \geq 2s-u$.
\end{lemma}

\begin{proof}
    By reverse induction, it is trivially true for $s=u$. Assume we have $\card {A_{s+1}\inter A_u} \geq 2(s+1)-u$ for some $s$ such that $k\leq s\leq u-1$. 
    
    We have that, by distribution of the intersection over the union operator:
    $$
    A_{s+1}\inter A_u=\left((A_{s+1}\inter A_s)\union (A_{s+1}\setminus A_s)\right)\inter A_u=\left(A_{s+1}\inter A_s\inter A_u\right)\union\left((A_{s+1}\setminus A_s)\inter A_u\right)
    $$
    Hence:
    $$
    \card {A_s \inter A_u}\geq \card {A_{s+1}\inter A_s\inter A_u} \geq \card {A_{s+1} \inter A_u} - \card {(A_{s+1}\setminus A_s)\inter A_u}\geq 2(s+1)-u-2=2s-u
    $$
Where we used that $\card {(A_{s+1}\setminus A_s)\inter A_u}\leq \card {A_{s+1}\setminus A_s}\leq 2$, because we removed at most two elements from $A_{s+1}$ to get $A_s$.
\end{proof}

\largesetoft*

\begin{proof}
    It is a direct implication of Corollary~\ref{nicelyordered} and Lemma~\ref{biginter}. 
    
    Indeed, $\card I =  \card {\{(a,t)\in A_s: \expo t u \leq t+1\}}$, and we have $\{(a,t)\in A_s: \expo t u \leq t+1\} \supseteq \{(a,t) \in A_s \inter A_u: \expo t u \leq t+1\} = \{(a,t)\in A_s \inter A_u\}=A_s\inter A_u,$ where the penultimate equality holds by Corollary~\ref{nicelyordered} applied to $u$. Lemma~\ref{biginter} states that $\card{A_s\inter A_u} \geq 2s-u$. 
\end{proof}

\begin{lemma}\label{atmostndegree}
    Let $\Delta_s=\expo t s-\expo t {s-1}$. Then, for every $s \in \{k+1, \hdots, n\}$:
    $$
     \sum_{u=s}^{\min\{2s-k-1, n\}}\left(2s-k - u\right)\Delta_{u} \leq n
     $$
\end{lemma}

An illustration of the proof can be found in Figure~\ref{fig:atmostn}.

\begin{figure}
    \centering
    \includegraphics[width=\linewidth]{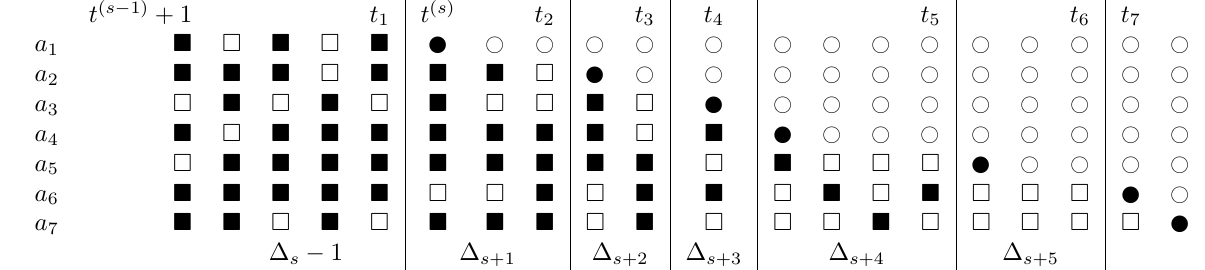}
    \caption{Illustration of proof of Lemma~\ref{atmostndegree}, with $s=7$ and $k=2$. Each column represents a round, each row represents an element $(a_i, t_i)$ of $A_s$, ordered according to $t_i$. Remember that every $t_i$ directly precedes a $t^{(u)}$. We add a square if $t\leq t_i$, and a circle otherwise. The entry is black if $\card{\I t {t_i} {a_i}}>\card{\I {t+1} {t_i} {a_i}}$ and white otherwise.
    We have one black circle per row. Indeed, we have that $\card {\I {t_i+1} {t_i} {a_i}}=\card {\{a_i\}} = 1 > 0 = \card {\I {t_i+2} {t_i} {a_i}}$. 
    Lemma~\ref{largesetoft} gives a lower bound on the number of squares for each column. The Strict Increments Lemma asserts that there are at most $k$ white squares per column. A lower bound for $S_s$ is thus the number of black entries.
    }
    \label{fig:atmostn}
\end{figure}

\begin{proof}
Let $s \in \{k+1, \hdots, n\}$. If for every $u$ such that $s\leq u\leq \min\{2s-k-1, n\}$, we have $\Delta_u = 0$, then the result is immediate. If not, let $x=\min\{u: s\leq u\leq \min\{2s-k-1, n\} \land \Delta_u >0\} $.
In the rest of the proof, 
let $u$ be restricted to fulfill $x\leq u\leq \min\{2s-k-1, n\}$ and let $I_s^u=\{i \in [s]: \expo t u \leq t_i+1\}$. By Lemma~\ref{largesetoft}, we have that $\card{ I_s^u} \geq 2s-u$.
    To prove the claim we will give upper and lower bounds for $S_s = \sum_{i \in [s]} \card{\I {\expo t {s-1}+1} {t_i} {a_i}}$. The upper bound directly follows from the fact that $\expo t {s-1}+1$ is the smallest round at which $A_s$ is strict, which implies that $S_s\le n$. 
    
    We next fix a round $t$ with $\expo t {u-1}+1\leq t\leq \expo t u$. 
    We have that $\expo t {s-1}$ is the largest round $A_s$ is not strict at (by Definition~\ref{thesets}), so it is strict at $t\geq \expo t {u-1}+1 > \expo t {s-1}$. By the Strict Increments Lemma (Lemma~\ref{bigin}),
    there exists a set $J_s^u(t)\subseteq I_s^u$ with $\card {J_s^u(t)} \geq \card {I_s^u} -k \geq 2s-u-k$ such that for every $j \in J_s^u(t)$, $\card {\I {t-1} {t_j} {a_j}} > \card {\I {t} {t_j} {a_j}}$. 
    As every $j \in J_s^u(t)$ belongs to 
    $I_s^u$ it follows that $t \leq \expo t u \leq t_j+1$. 
    Thus, for every $(j,t) \in [s]\times \N$ with $t>t_j+1$ it holds that it does not belong to $J_s^u(t)$,
    or, equivalently, $\1(j\in J_s^u(t))=0$.
    
    We can then write:
    $$
         S_s:=\sum_{i \in [s]} \card{\I {\expo t {s-1}+1} {t_i} {a_i}}=
         \sum_{i\in[s]}\left(\card{\I {t_i+1} {t_i} {a_i}}+ \sum_{t=\expo t {s-1}+1}^{t_i}\left( \card{\I {t} {t_i} {a_i}}-\card{\I {t+1} {t_i} {a_i}}\right)\right)$$

         We first separate the second sum according to which interval $[\expo t {u-1}+1, \expo t u]$ round $t$ belongs to, and furthermore add a third sum that is equal to $0$, since $\1(j\in J_s^u(t))=0$ for every $(j,t) \in [s]\times \N$ such that $t>t_j+1$.
         For the consecutive inequality note that by the definition of $J_s^u(t)$, we have that in the second sum $\card{\I {t} {t_i} {a_i}}-\card{\I {t+1} {t_i} {a_i}} \geq \1(i\in J_s^u({t+1}))$.
         \noindent\resizebox{\linewidth}{!}{
    \parbox{\linewidth}{
         \begin{align*}
         S_s&\geq
         \sum_{i\in[s]}\left(\card{\I {t_i+1} {t_i} {a_i}}
         + \sum_{t=\expo t {s-1}+1}^{t_i}\sum_{u=s}^{\min\{2s-k-1, n\}}\1(\expo t {u-1}+1\leq t+1\leq \expo t u)\left( \card{\I {t} {t_i} {a_i}}-\card{\I {t+1} {t_i} {a_i}}\right)\right.
         \\&\quad\quad\quad\quad\quad\quad\quad\quad\quad+\left.\sum_{t=t_i+2}^{\expo t {\min\{2s-k-1, n\}}}\sum_{u=s}^{\min\{2s-k-1, n\}}\1(\expo t {u-1}+1\leq t\leq \expo t u)\1(i\in J_s^u(t))\right)\\
         &\geq
         \sum_{i\in[s]}\left(\card{\I {t_i+1} {t_i} {a_i}}+ \sum_{t=\expo t {s-1}+1}^{t_i}\sum_{u=s}^{\min\{2s-k-1, n\}}\1(\expo t {u-1}+1\leq t+1\leq \expo t u)\1(i\in J_s^u({t+1}))\right.
         \\&\quad\quad\quad\quad\quad\quad\quad\quad\quad\left.+\sum_{t=t_i+2}^{\expo t {\min\{2s-k-1, n\}}}\sum_{u=s}^{\min\{2s-k-1, n\}}\1(\expo t {u-1}+1\leq t\leq \expo t u)\1(i\in J_s^u(t))\right) \\
         &=
        \sum_{i\in[s]}\left(\card{\I {t_i+1} {t_i} {a_i}}+ \sum_{t=\expo t {s-1}+2}^{t_i+1}\sum_{u=s}^{\min\{2s-k-1, n\}}\1(\expo t {u-1}+1\leq t \leq \expo t u)\1(i\in J_s^u({t}))\right.
        \\&\quad\quad\quad\quad\quad\quad\quad\quad\quad\left.+\sum_{t=t_i+2}^{\expo t {\min\{2s-k-1, n\}}}\sum_{u=s}^{\min\{2s-k-1, n\}}\1(\expo t {u-1}+1\leq t\leq \expo t u)\1(i\in J_s^u(t))\right) \end{align*}}}
        \noindent\resizebox{\linewidth}{!}{
    \parbox{\linewidth}{
         \begin{align*}
         S_s&\geq
         \sum_{i\in[s]}\Biggl(\card{\I {t_i+1} {t_i} {a_i}}\\&\quad\quad\quad\quad+ \sum_{t=\expo t {s-1}+2}^{\max \{ t_i+1, \expo t {\min\{2s-k-1, n\}}\}}\sum_{u=s}^{\min\{2s-k-1, n\}}\1(\expo t {u-1}+1\leq t\leq \expo t u)\1(i\in J_s^u({t}))\Biggr) =: X_1
         \end{align*}}}
         
         We next invert the two sum symbols, and delete terms where $\1(\expo t {u-1}+1\leq t\leq \expo t u)=0$:
         \resizebox{\linewidth}{!}{\parbox{\linewidth}{
         \begin{align*}
         X_1 &\geq
         \sum_{i \in [s]}\left(1+\sum_{t=\expo t {x-1}+2}^{\expo t x}\1(i \in J_s^x(t))+\sum_{u=x+1}^{\min\{2s-k-1, n\}}\sum_{t=\expo t {u-1}+1}^{\expo t u}\1(i \in J_s^u(t))\right)\\
         &\geq
        s+\sum_{t=\expo t {x-1}+2}^{\expo t x}\sum_{i \in [s]}\1(i \in J_s^x(t))+\sum_{u=x+1}^{\min\{2s-k-1, n\}}\sum_{t=\expo t {u-1}+1}^{\expo t u}\sum_{i \in [s]}\1(i \in J_s^u(t)) =: X_2\\
        \intertext{Using that $\sum_{i \in [s]}\1(i \in J_s^u(t)) = \card{J_s^u(t)} \geq 2s-k - u$ and $\Delta_u=\expo t u - \expo t {u-1}$:}
         X_2 &\geq s+(\Delta_x-1)(2s-k-x)+\sum_{u=x+1}^{\min\{2s-k-1, n\}}(2s-k-u)\Delta_u
         \geq \sum_{u=s}^{\min\{2s-k-1, n\}}(2s-k-u)\Delta_u,
    \end{align*}}}
    where the last inequality follows since   $x \geq s$ which implies that $2s-k-x \leq 2s-k-s \leq s$ and that $\Delta_u = 0$ for $u < x$.

\noindent
As $\sum_{i \in [s]} \card{\I {\expo t {s-1}+1} {t_i} {a_i}}=S_s \leq n$, the claim follows.   
\end{proof}

We now recall the definition of the strict rounds graph, which is built to harness the previous result.

\strictrounds*

The graph is represented in \figref{fig:strict}. Each node $s$ with $k+1 \le s \le (n+k)/2$
has $s-k+ 2$ incoming edges, namely $(s,s)$, $(s+1, s), \dots, (2s-k-1,s)$,
with weights $s-k, s-k-1, \dots, 1$, respectively. Each node
$s$ with $(n+k)/2 < s \le n$ has 
$n+1-s$ incoming edges, namely $(s,s)$, $(s+1, s), \dots, (n,s)$,
with weights $s-k, s-k-1, \dots, 2s-k-n$, respectively. 

Each node $s$ has $s+1 - \floor{\frac{s+k}{2}}$ outgoing edges, namely $(s,s), (s, s-1), \dots, (s,\floor{\frac{s+k}{2}} + 1)$ of weight
$s-k, s-k - 2, \dots$, respectively. If $s-k$ is even, all outgoing edges of $s$ have even weight and the edge
$(s,\floor{\frac{s+k}{2}} + 1)$ has weight 2.
If $s-k$ is odd, all outgoing edges of $s$ have odd weight and
the edge $(s,\floor{\frac{s+k}{2}} + 1)$  has weight 1.

\begin{lemma}\label{lem:strict}
    Let $u\in [n]$, and define $E_u^{out}:=\{(s,v)\in E:s\leq u\}$. In the strict rounds graph, we have that 

    $$
    \sum_{(s,v) \in E_u^{out}} \Delta_s w(s,v) \leq (u-k) \cdot n
    $$
\end{lemma}

\begin{proof}
    We are summing over all the out-edges of all the vertices with a label at most $u$. 
    Let $E_u^{in}=\{(s,v)\in E:v\leq u\}$. 
    Every outgoing edge of a node $s \le u$ must end at a node $v$ with $v \le s \le u$. Thus, it is contained in the set of incoming edge of all nodes $v$ with $v \le u$. Said differently, $E_u^{out} \subseteq E_u^{in}$.
    
    Since all the terms are positive, we have:

    \begin{multline*}
        \sum_{(s,v) \in E_u^{out}} \Delta_s w(s,v)  \leq\sum_{(s,v) \in E_u^{in}} \Delta_s w(s,v)  = \sum_{v \leq u}\sum_{s:(s,v) \in E}\Delta_s w(s,v) \\\leq \sum_{v \leq u}\sum_{s:(s,v) \in E}\Delta_s (2v-k-s)\leq \sum_{v \leq u}\sum_{s = v}^{\min\{n, 2v-s-1\}}\Delta_s (2v-k-s)
    \end{multline*}

    Applying Lemma~\ref{atmostndegree},
    the claim follows.
\end{proof}

\begin{restatable}{lemma}{alphas}\label{alphas}
    Let $s \geq k+1$. If $s-k$ is odd, we have $\sum_{v:(s,v) \in E}w(s,v)={(\frac{s-k+1} 2)^2}$. If it is even, we have $\sum_{v:(s,v) \in E}w(s,v)={(\frac{s-k} 2)^2+\frac{s-k} 2}$.
\end{restatable}

Basically, as discussed above and seen in \figref{fig:strict}, if $s-k$ is odd, we have that $\sum_{v:(s,v) \in E}w(s,v)=1+3+5+\cdots +s-k={(\frac{s-k+1} 2)^2}$, while if it is even, $\sum_{v:(s,v) \in E}w(s,v)=2+4+6+\cdots+s-k={(\frac{s-k} 2)^2+\frac{s-k} 2}$. The full proof can be found in Appendix~\ref{appendix}.

\begin{definition}
    We define the following numbers: $\beta=\frac {\pi^2 +6} 6, \alpha_s={(\frac {s-k+1} 2)^2}$ if $s-k$ is odd, $\alpha_s={(\frac{s-k} 2)^2+\frac{s-k} 2}$ if $s-k$ is even.
\end{definition}
Note that $\alpha_s$ is an integer for all values of $s \ge k$.
By applying this definition to the bounds of Lemmata~\ref{lem:strict} and~\ref{alphas} we achieve the following result.

\corsum*

\littledeltas*

\begin{proof}
    We will show that $\sum_{s=k+1}^{n}\delta_s \alpha_s$  is maximized when $\delta_s=\alpha^{-1}_s\cdot n$ for every $s$. Let $\{\delta_s\}_{k+1\leq s \leq n}$ be such that
    $\sum_{s=k+1}^{n}\delta_s \alpha_s$ if maximized.
    Assume by contradiction that there exists a $k+1\leq s\leq n$ such
    that $\delta_s < \alpha^{-1}_s \cdot n$. Let $v$ be the smallest such $s$ and set $\epsilon = n - \delta_v \alpha_v>0$. 
    We can then build a different solution $\gamma$ by setting $\gamma_s=\delta_s$ for every $s \notin \{v, v+1\}$, and setting $\gamma_v=\delta_v+{\epsilon}{\alpha^{-1}_v} = \alpha^{-1}_v \cdot n > \delta_v$ and $\gamma_{v+1}=\delta_{v+1}-{\epsilon}{\alpha^{-1}_{v+1}}$. 
    Note that for every $u\in \{k+1, \hdots n\}$, $    \sum_{s \leq u} \gamma_s \alpha_s \leq (u-k)\cdot n $ as 
    
    (1) for $u < v$ 
    it holds that $\sum_{s \leq u} \gamma_s \alpha_s = (u-k)\cdot n $,
    
     (2) for $u = v$ it holds that
    $    \sum_{s \leq u} \gamma_s \alpha_s = 
       (\epsilon \alpha^{-1}_v)\alpha_v +  \delta_v \alpha_v  + \sum_{s < u} \delta_s \alpha_s =
       (n - \delta_v \alpha_v)  + \delta_v \alpha_v +  \sum_{s < u} \delta_s \alpha_s = 
       n + (u-1-k) \cdot n  = (u-k) \cdot n$, and
       
       (3) for $u \ge v+1$ if holds that 
       $    \sum_{s \leq u} \gamma_s \alpha_s = 
       \sum_{s \leq v-1} \delta_s \alpha_s + \gamma_{v} \alpha_v + \gamma_{v+1}\alpha_{v+1}  + \sum_{v+2 \le s \leq u} \delta_s \alpha_s = 
       \sum_{s \leq v - 1} \delta_s \alpha_s + 
       \delta_{v} \alpha_{v} + \epsilon \alpha_{v+1} \alpha_{v+1}^{-1} +
       \delta_{u} \alpha_{u}^{-1} - \epsilon \alpha_{u} \alpha_{u}^{-1} 
       + \sum_{v+2 \le s \leq u} \delta_s \alpha_s
       =
        \sum_{s \le u} \delta_s \alpha_s \le (u-k) \cdot n,$ and

    Now note that
    $\sum_{s=k+1}^{n}\gamma_s \alpha_s > \sum_{s=k+1}^{n}\delta_s \alpha_s$ as the $\alpha_s$ values are decreasing in $s$, which gives a contradiction to the assumption that $\delta$ maximizes $\sum_{s=k+1}^{n}\delta_s \alpha_s$. It now follows that
    

    $$\sum_{s=k+1}^{n} \delta_s \le \sum_{s=k+1}^{n}\alpha^{-1}_s n\leq\sum_{s-k=1}^{\infty}\alpha^{-1}_s n 
    \le n\sum_{\ell\in \N}\frac{1}{\ell^2}
    +n\sum_{\ell \in \N}\frac{1}{\ell^2+\ell}=\beta n.
    $$
    where the last equation follows by the fact that $\sum_{\ell\in \N}\frac{1}{\ell^2} = \pi^2/6$ and $\sum_{\ell\in \N}\frac{1}{\ell^2+\ell}=\sum_{\ell\in \N}\frac{1}{\ell}-\frac 1 {\ell+1}=1$.
\end{proof}

\atmostbeta*

\bigtheorem*

\begin{proof}
    We have $\sum_{s=k+1}^{n}\Delta_s=\sum_{s=k+1}^n \expo t s-\expo t {s+1}=t'-\expo t k$, which is at most  $\frac{\pi^2+6}{6} n$ by Corollary~\ref{corollary}. 
    Setting $\expo t k=1$, we get that $t'\leq\frac{\pi^2+6}{6}n+1$. Lemma~\ref{cover} ensures that $A_k$ is a cover of $[n]$ for $t'$, which means that for every $x \in [n]$, there exists a $i \in [k]$ such that $x \in \Out {t_i+1} {t'} {a_i}$. 
    Since $\min_i\{t_i\} = \expo t k-1=0$, this means that every $x \in  [n]$ belongs to $\Out {1} {t'} {a_i}$, which implies that $t^c_k(\F_n^k)\leq t'$.
\end{proof}

\subsection{The Lower Bound}

We build a lower bound example based on the lower bound for broadcasting on trees. A figure and analysis of that example can be found in Appendix~\ref{appendix}, which yields the following result:

\begin{restatable}{theorem}{lowerbthree}
    $t^c_k(\RT_n^k) \geq \ceil{\frac {3n-3k}{2} -1}$
\end{restatable}

\section{\texorpdfstring{$k$-Broadcasting on $k$-rooted Networks}{k-Broadcasting on k-rooted Networks}}\label{multiroot}
In this section, we study the case where the adversary has to choose a communication network that has $k$ roots at least in each round. This allows us to enforce not only broadcast, but rather $k$-broadcast. 

The problem is very similar to the one studied in the Section~\ref{tree}, and indeed a slight modification of the proofs there work for this problem.

\subsection{The Upper Bound}

By using a technique very similar to Section~\ref{tree}, we will show that whenever we have a set $A\subset [n]$ of size at most $k-1$, we can find a process $p \notin A$ that has broadcast, as long as we have waited for a large enough number of rounds. 

\begin{definition}
    Let $t$ be a round. We denote by $R_t$ the set of the roots of $G_t$. 
\end{definition}
Recall that $\card {R_t} \geq k$.

The following two lemmata, very similar to Section~\ref{tree}, can be proven by looking at paths going from a root to a carefully chosen node.

\begin{restatable}{lemma}{ingetsbigger}\label{ingetsbigger}
    Let $t\leq t'$ be rounds, and let $x$ and $r$ be processes such that, $r \in R_t$ and $r \notin \I {t+1} {t'} x$. Then $\card {\I t {t'} x}> \card {\I {t+1} {t'} x}$. 
\end{restatable}

\begin{proof}
    We will show that there exists a process $y \in \I {t+1} {t'} x$ that has an in-neighbor $z$ in $G_t$ such that $z \notin \I {t+1} {t'} x$. Then, by Transitivity (Lemma~\ref{inin}), $z \in \I t {t'} x$. By Monotonicity (Lemma~\ref{trivialinclusion}), we have that $\I {t+1} {t'} x \subset \I {t} {t'} x$, this will show that $\card {\I t {t'} x}> \card {\I {t+1} {t'} x}$.

    Let us now find such a $z$. Consider the path from $r$ to $x$ in $G_t$. Since $r \notin \I {t+1} {t'} x$, and trivially $x \in \I {t+1} {t'} x$, this path must include an edge $(z,y)$ such that $z \notin \I {t+1} {t'} x$, $y \in \I {t+1} {t'} x$.
\end{proof}

\begin{restatable}{lemma}{outgetsbigger}\label{outgetsbigger}
    Let $t\leq t'$ be rounds, and let $x$ and $r$ be processes such that $r \in R_{t'}$, and $r \in \Out {t} {t'-1} x$.
Then $\card {\Out t {t'} x}> \card {\Out {t} {t'-1} x}$, unless $\card {\Out {t} {t'-1} x}=n$.
\end{restatable}

\begin{proof}
    Let us look at the case $\card {\Out {t} {t'-1} x}<n$. We will show that there exists a process $y \in \Out {t} {t'-1} x$ that has an out-neighbor $z$ in $G_{t'}$ such that $z \notin \Out {t} {t'-1} x$. Then, by Transitivity (Lemma~\ref{outout}), $z \in \Out t {t'} x$. By Monotonicity, we obviously have $\Out {t} {t'-1} x \subset \Out {t} {t'} x$, which implies that $\card {\Out t {t'} x}> \card {\Out {t} {t'-1} x}$.

    Let us now find such a $z$. Since $\card {\Out {t} {t'-1} x}<n$, there exists a process $a$ such that $ a \notin \Out {t} {t'-1} x$. Consider the path from $r$ to $a$ in $G_{t'}$. Since $r \in \Out {t} {t'-1} x$, and $a \notin \Out{t} {t'-1} x$, this path must include an edge $(y,z)$ such that $z \notin \Out {t} {t'-1} x$, $y \in \Out {t} {t'-1} x$.
\end{proof}

We now give a lemma that links the number of in-neighbors of a node to the number of previous rounds whose root are in-neighbors of said node:

\begin{restatable}{lemma}{manyones}\label{manyone}
    Let $x$ be a process, let $t_1\leq t_2$ be rounds, and let $r_t$ denote a root of $R_t$ for every $t \in \{t_1, \hdots, t_2\}$, then it holds that
    $$
    \card{\{t: t_1\leq t\leq t_2, r_t \notin \I {t_1} {t_2} x\}}+1 \leq \card{\I {t_1} {t_2} x}
    $$
\end{restatable}

\begin{proof}
    Let $A=\{t: t_1\leq t\leq t_2, r_t \notin \I {t_1} {t_2} x\}$. Then, in particular, for any $ t \in A$, we have $r_t \notin \I {t+1} {t_2} x$. Then, for all $t \in a$, applying Lemma~\ref{ingetsbigger}, we have $\card{\I {t} {t_2} x} > \card{ \I {t+1} {t_2} x }$. Let $k= |A|$ and let $A=\{t^{(1)}, \hdots, t^{(k)}\}$, with $t^{(i)}<t^{(i+1)}$ for $1\leq i<k$. Then by Monotonicity (Lemma~\ref{trivialinclusion}) it follows that
    $$
    \card{\I {t_1} {t_2} x} \geq \card{\I {t^{(1)}} {t_2} x } >\card{\I {t^{(1)}+1} {t_2} x } \geq \card{\I {t^{(2)}} {t_2} x }>\card{\I {t^{(2)}+1} {t_2} x }\geq \hdots >\card{\I {t^{(k)}+1} {t_2} x }\geq 1
    $$

    There are $k=\card{A}$ inequalities over integers, which concludes the proof.
\end{proof}

We now define the rounds graph avoiding some set $A$, which is the equivalent of the rounds graph of Section~\ref{tree}, this time being very careful not to choose any process from $A$ being used as a root for a round.

\begin{definition}
    Let $A$ be an arbitrary set of processes with $\card{A}<k$. We define the \emph{rounds graph avoiding set $A$} as follows:
    It contains $2n+\ceil{\sqrt2 n}+\card A$ nodes, namely
    \begin{enumerate}
        \item one {\em process node} representing each process.
        \item one {\em round node} for each of the first $\ceil{(1+\sqrt{2})n}+\card A$ rounds.
    \end{enumerate}
    For each round $t$, let $r_t$ be a vertex of $ R_t\setminus A$. The graph contains
    \begin{enumerate}
        \item a directed edge to every round node $t$ from each process node $p$ with $p \in \I 1 {t-1} {r_t}$.
        \item  a directed edge from every round node $t<\ceil{\sqrt 2 n}+\card A$ to a round  node $t'>t$ if $r_t \in \I 1 {t'-1} {r_{t'}}$.
    \end{enumerate}
\end{definition}

We now find a node in the rounds graph that has out-degree $n$, that is not a node from $A$. This node will represent a process that has broadcast.

\begin{restatable}{lemma}{pigeonhole}\label{lem:pig2}
    For any set $A$ of processes with $\card{A} < k$, the rounds graph avoiding set $A$ contains at least one node of out-degree $n$ that is not a node representing a process of $A$.
\end{restatable}

\begin{proof}
    As in the proof of Lemma~\ref{lem:pig1} we apply the pigeonhole principle on the edges, however, this time we will ignore out-edges from nodes representing processes from $A$.
    Let us look at round $t$. If $t\leq\ceil{\sqrt 2 n}+\card A$, then the round node $t$ has in-degree at least $t$. Indeed, it has in-degree:
    \begin{multline*}
    \card{\{t': 1\leq t'\leq t-1, {r_{t'} }\in  \I {1} {t-1} {r_t}\}}+ \card{\I {1} {t-1} {r_t}}\\ \geq 1+\card{\{t': 1\leq t'\leq t-1, {r_{t'} }\in \I {1} {t-1} {r_t}\}}+\card{\{t': 1\leq t'\leq t-1, {r_{t'} }\notin \I {1} {t-1} {r_t}\}}=t
    \end{multline*}
    where we used Lemma~\ref{manyone} for the inequality. Therefore, it has at least $t-\card A$ in-neighbors not in $A$.
    Similarly, if $t>\ceil{\sqrt 2 n}+\card A$, then the round node $t$ has in-degree at least $\ceil{\sqrt 2 n}+\card A$. Indeed, it has in-degree:
    \begin{multline*}
    \card{\{t': 1\leq t'< \ceil{\sqrt 2 n}+\card A, {r_{t'} }\in  \I {1} {t-1} {r_t}\}}+ \card{\I {1} {t-1} {r_t}}\\ \geq 1+\card{\{t': 1\leq t'< \ceil{\sqrt 2 n}+\card A, {r_{t'} }\in \I {1} {t-1} {r_t}\}}+\card{\{t': 1\leq t'\leq t-1, {r_{t'} }\notin \I {1} {t-1} {r_t}\}}\\
    \geq 1+\card{\{t': 1\leq t'< \ceil{\sqrt 2 n}+\card A, {r_{t'} }\in \I {1} {t-1} {r_t}\}}\\+\card{\{t': 1\leq t'<\ceil{\sqrt 2 n}+\card A, {r_{t'} }\notin \I {1} {t-1} {r_t}\}}=\ceil{\sqrt 2 n}+\card A
    \end{multline*}
    where we used Lemma~\ref{manyone} for the first inequality. Therefore, it has at least $\ceil{\sqrt 2 n}$ in-neighbors not elements of $A$.

    Summing the in-degrees over all the rounds, we get that the number of edges $\card{E_{\overline{A}}}$ with no endpoint in $A$ is at least:

    $$
    \card{E_{\overline{A}}}\geq \sum_{t=\card A}^{\ceil{\sqrt 2 n}+\card A}(t-\card A)+n\times \ceil{\sqrt 2 n} = \frac {\ceil{\sqrt 2 n}\left(\ceil{\sqrt 2 n}+1\right)} 2 + n \ceil{\sqrt 2 n}> \ceil{(1+\sqrt 2) n}n
    $$

    But only the $n-\card A$ nodes representing the processes and the nodes representing the first $\ceil{\sqrt 2 n}-1+\card A$ rounds have out-edges among those counted. The pigeonhole principle asserts then that one of those nodes has an out-degree of at least $n$.
\end{proof}

\begin{lemma}\label{manybroadcast}
    For any set $A$ of processes with $\card A <k$, there exists a process $z \notin A$ that has broadcast its identifier to everyone after $\ceil{(1+\sqrt 2)n}+\card A$ rounds.
\end{lemma}

\begin{proof}
        Build the rounds graph avoiding $A$. We know that in that graph, there is a node $y$ of degree at least $n$, that is not a process node representing $A$. Define $z$ as follows: (1) If $y$ represents a process, let $z$ be that process. It follows that $z \notin A$. (2) If $y$ represents a round $t$, let $z$ be the root $r_t$, which is not in $A$ by definition of the graph. We will show in either case that $z$ must have broadcast before $\ceil{(1+\sqrt 2)n}+\card A$ rounds.

    Let $t_1< \hdots < t_n$ be the rounds $y$ has out-edges to. By definition, and in both cases (1) and (2), this means $z\in \I 1 {t_i-1} {r_{t_i}} $
    which is equivalent to $ r_{t_i} \in \Out 1 {t_i-1} z$ for every $i \in [n]$. 
   By Lemma~\ref{outgetsbigger}, we thus have, for every $i \in [n]$, that $\card{\Out 1 {t_i} z} >\card{ \Out 1 {t_i-1} z}$. Then:

    $$
    \card{\Out 1 {t_n} z} >\card{ \Out 1 {t_n-1} z}\geq \card{\Out 1 {t_{n-1}} z} >\card{ \Out 1 {t_{n-1}-1} z}\geq \hdots \geq\card{\Out 1 {t_1} z} >\card{ \Out 1 {t_1-1} z}\geq 0
    $$

    We have $n$ strict inequalities over non-negative integers, the largest one must be at least $n$, which implies that $z$ has broadcast.
\end{proof}

This naturally leads to an upper bound for the $k$-broadcasting on $k$-rooted networks:

\begin{theorem}
    $t^*_k(\R_n^k) \leq \ceil{(1+\sqrt 2)n}+k-1$
\end{theorem}

\begin{proof}
    By contradiction, assume that the set $A$ of elements that have broadcast after $\ceil{(1+\sqrt 2)n}+k-1$ rounds is smaller than $k$ and apply Lemma~\ref{manybroadcast} to find a process not in $A$ that has broadcast in time less than $\ceil{(1+\sqrt 2)n}+\card A$ rounds, contradicting the assumption.
      
\end{proof}

\subsection{The Lower Bound}

We build a lower bound example based on the lower bound for broadcasting on trees. A figure and analysis of that example can be found in Appendix~\ref{appendix}, which yields the following result:

\begin{restatable}{theorem}{lowerbtwo}
    $t^*_k(\R_n^k) \geq \ceil{\frac {3n-9k}{2} +2}$
\end{restatable}

\section{Related Work}\label{sec:relwork}

Broadcasting, gossiping, and other information dissemination problems have been studied by the distributed computing community for decades already~\cite{hedetniemi1988survey}. Most classic literature on network broadcast considers a static setting, e.g., where in each round each node can send information to one neighbor~\cite{hromkovivc1996dissemination}. This model has also been explored in the context of gossiping, e.g., by Fraigniaud and Lazard~\cite{fraigniaud1994methods}. Kuhn, Lynch and Oshman~\cite{kuhn2010distributed} explore the all-to-all data dissemination problem (gossiping) in an undirected dynamic network, where processes do not know beforehand the total number of processes and must decide on that number. Ahmadi, Kuhn, Kutten, Molla and Pandurangan~\cite{DBLP:conf/icdcs/AhmadiKKMP19} study the message complexity of broadcast in an undirected dynamic setting, where the adversary pays up a cost for changing the network. Broadcast has also been studied in dynamic communication networks which evolve randomly, e.g., by Clementi et al.~\cite{clementi2013rumor}, and in the radio network model~\cite{ellen2021constant}, just to give a few examples. 

A closely related yet different problem to broadcasting is the consensus problem. Our model builds up on the heard-of model first introduced by Charron-Bost and Schiper~\cite{charron2009heard}, where authors prove results for the solvability of consensus also considering oblivious message adversaries. Among other results, they give a $\log n$ upper bound for nonsplit graphs, which are graphs for which every pair of nodes has a common in-neighbor. This would result in an $n\log n$ upper bound for rooted trees when combining it with the result of Charron-Bost, F{\"u}gger and Nowak~\cite{charron2015approximate}. F{\"u}gger, Nowak, and Winkler~\cite{fugger2020radius} prove that the time complexity of broadcast is a lower bound for consensus time. A general characterization of oblivious message adversaries on which consensus is solvable, based on broadcastability, has been presented by Coulouma, Godard and Peters in~\cite{coulouma2015characterization}. A time complexity analysis has further been studied by Winkler, Rincon Galeana, Paz, Schmid, and Schmid~\cite{itcs23consensus}. Another similar problem is agreement, considered by Santoro and Widmayer~\cite{santoro1989time}, where only a $k$-majority should agree on a value, as opposed to everyone for consensus. Afek, Gafni, Rajsbaum, Raynal and Travers~\cite{AfekGRRT10} studied a generalization to consensus that is $k$-set consensus, where each node has to decide on a value such that all the nodes together do not decide on more than $k$ different values. 

In this paper, we have studied the broadcasting problem on directed dynamic networks, with an adversary that can choose the communication network at each round among rooted trees. Zeiner, Schwarz, and Schmid~\cite{schwarz2017linear} give a $n\log n$ upper bound to our exact problem by using graph-theoretic reasoning. They also give a $\ceil{\frac{3n-1}{2}}-2$ lower bound by providing an explicit example. They further show that under an adversary that can only choose rooted trees with a fixed number of leaves or internal nodes, broadcast time is linear.

There has also been interest in a problem variant which only differs in the pool of networks the adversary can choose a network from for each communication round. F{\"u}gger, Nowak, and Winkler~\cite{fugger2020radius} give an $O( \log\log n)$ upper bound if the adversary can only choose nonsplit graphs. Combined with the result of Charron-Bost, F{\"u}gger, and Nowak~\cite{charron2015approximate} that states that one can simulate $n-1$ rounds of rooted trees with a round of a nonsplit graph, this gives the previous $O(n\log\log n)$ upper bound for broadcasting on trees. Dobrev and Vrto~\cite{dobrev2002optimal, dobrev1999optimal} give specific results when the adversary is restricted to hypercubic and tori graphs with some missing edges.

\noindent \textit{Bibliographic note: an announcement of this work has been presented at ACM PODC 2022~\cite{podc2022broadcasting}}.

\section{Conclusion}\label{sec:conclusion}

In this paper, we considered an innovative version of the classic broadcast problem where processes communicate across a dynamically changing network, as it often arises in practice (e.g., due to interference). Like in the static setting, the broadcast problem on dynamic networks is related to consensus and leader election: broadcast is a prerequisite for consensus, and hence, the time complexity of broadcast is a lower bound for the consensus and leader election complexity. 

Our main contribution is a proof that the broadcast time is at most linear in this setting, which is asymptotically optimal. We further presented several natural generalizations of our model and result. 

Our work opens several avenues for future research. In particular, it will be interesting to study the broadcast time also in non-adversarial environments where graphs evolve according to a random process (e.g., due to random node movements). It will also be interesting to further explore the implications of our methods on the closely related consensus problem.

\bibliographystyle{plain}
\bibliography{main}

\newpage
\appendix

\section{Omitted Proofs}
\label{appendix}
\subsection{Basic Tools}
\label{appendixtool}
\first*
\begin{proof}
    If $t>t'+1$, the equivalence is empty. If $t=t'+1,$ then $x \in \Out {t} {t'} y \Leftrightarrow x=y \Leftrightarrow y \in \I {t} {t'} x$. If $t\leq t'$, $x \in \Out {t} {t'} y \Leftrightarrow (y,x) \in \graph G t {t'} \Leftrightarrow y \in \I {t} {t'} x$.
\end{proof}

\outout*

\begin{proof}
    \emph{ii.} and \emph{iii.} will follow from \emph{i.} and Lemma~\ref{easy}, so we only prove \emph{i.}.

    Since $y \in \Out t {t'} x$, this means $\Out t {t'} x \neq \varnothing$. Therefore, $t\leq t'+1$. If $t=t'+1$, we have that $\Out t {t'} x = \{x\} \Rightarrow y=x$, and thus $z \in \Out t {t''} x$. We will then only need to consider the case $t\leq t'$.
    
    Similarly, since $z \in \Out {t'+1} {t''} y$, we have that $\Out {t'+1} {t''} y \neq \varnothing$ and thus $t'+1\leq t''+1$. If $t'=t''$, we have that $\Out {t'+1} {t''} y = \{y\}$ thus $z=y$ and therefore $z \in \Out t {t''} x$. Hence we only need to consider the case $t\leq t'\leq t''-1$.
    
    In that last case, we know that $(x,y)$ is an edge in $\graph G t {t'}$ and $(y, z)$ is an edge in $\graph G {t'+1} {t''}$. By definition of a product graph, this means that $(x,z)$ is an edge in $\graph G {t} {t''}$, and thus that $z \in \Out {t} {t''} x$.
\end{proof}

\trivialinclusion*
\begin{proof}
    Let us first consider the case $t_2 > t_3+1$. Then $\I {t_2} {t_3} x = \Out {t_2} {t_3} x = \varnothing$ and the result is trivial. Similarly, if $t_2=t_3+1$, then $\I {t_2} {t_3} x= \Out {t_2} {t_3} x= \{x\}$, and then either $t_1=t_4+1$ which means $\I {t_1} {t_4} x= \Out {t_1} {t_4} x= \{x\}$, or $t_1\leq t_4$, and then $x$ has a self-loop in $G_t$ for every $t_1\leq t\leq t_4$, which implies $x \in \Out {t_1} {t_4} x$ and $x \in \I {t_1} {t_4} x$. 
    
    Let's now consider the case $t_2 \leq t_3$. Let $y\in \I {t_2} {t_3} x$ (respectively $z \in \Out {t_2} {t_3} x$). Then edge $(y,x)$ ($(x,z)$) is in graph $\graph G {t_2} {t_3}$. Since all rounds have self-loops, we clearly have that edge $(y,y)$ ($(x,x)$) is in the graph $\graph G {t_1} {t_2}$. Similarly, edge $(x,x)$ ($(z,z)$) is in graph  $\graph G {t_3} {t_4}$. This shows that $(y,x)$ ($(x,z)$) is in graph $\graph G {t_1} {t_4}$ and thus $y \in \I {t_1} {t_4} x$ ($z \in \Out {t_1} {t_4} x$).
\end{proof}

\subsection{Broadcasting on Trees}
\label{appendixlow}
\subsubsection{The Lower Bound}

The lower bound for this problem, due to Zeiner, Schwarz and Schmid~\cite{schwarz2017linear} can be seen in \figref{fig:low1}. The sequence of graphs in that example do not achieve broadcast before $\ceil{\frac {3n-1} 2 }-2$ rounds. This examples proves the theorem:

\begin{figure}
    \centering
    \includegraphics{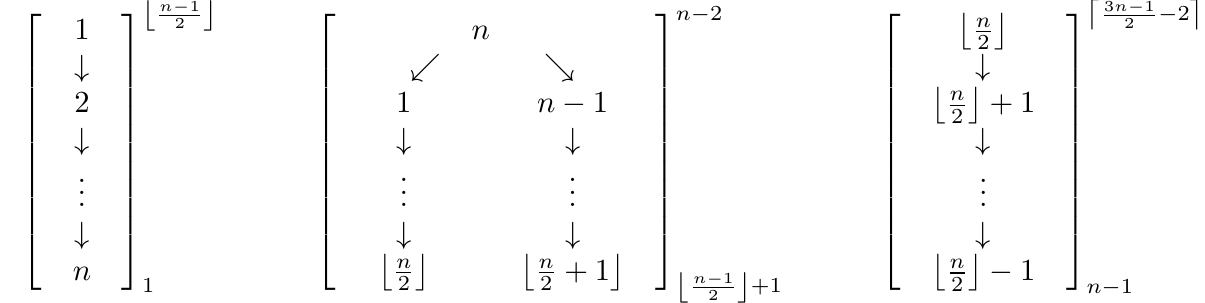}
    \caption{The lower bound example for the broadcasting on trees problem \cite{schwarz2017linear}. $\left[ G \right]_a^b$ means that network $G$ is the communication graph for all rounds between $a$ and $b$ (inclusive).}
    \label{fig:low1}
\end{figure}

\lowerb*

\subsection{\texorpdfstring{Covering on $k$-Forests}{Covering on k-Forests}}
\label{appendix2}
\subsubsection{The Upper Bound}

\begin{figure}
    \centering
    \includegraphics[width=0.7\linewidth]{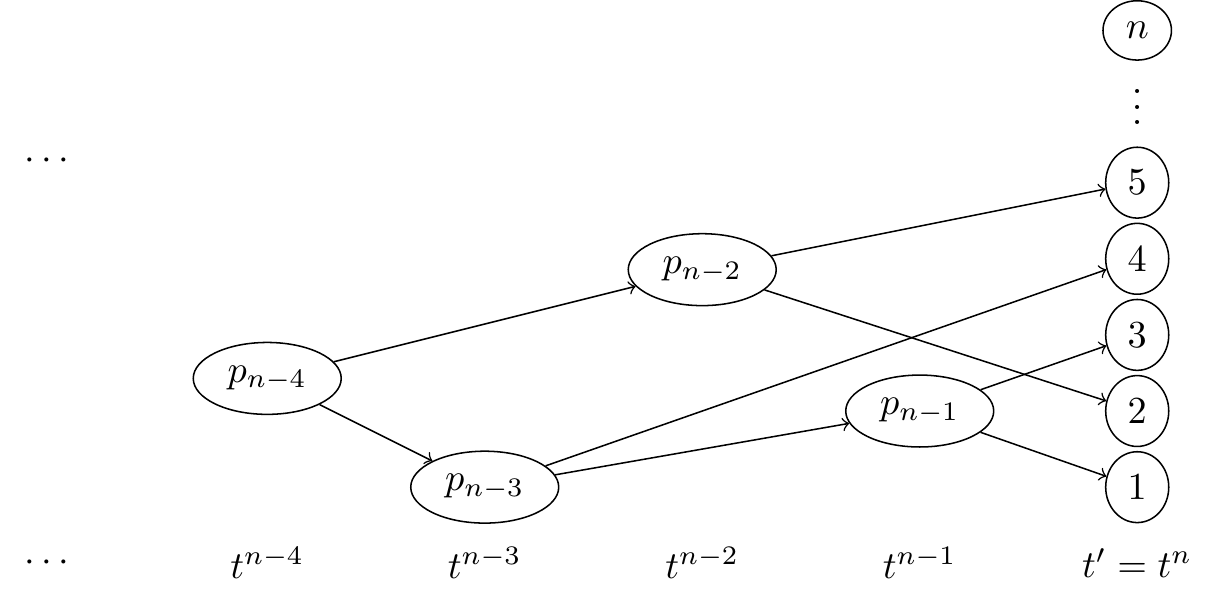}
    \caption{Main idea of the proof for the upper bound on covering on $k$-forests. An edge from process $x$ in column $t_1$ to process $y$ in column $t_2$ means that $x\in \I {t_1} {t_2-1} y$ (unless $t_2=t'$ where we don't require the offset). We start with a cover of size $n$ at $t'$ then go back in time, finding a smaller cover at every step. After the first step, we don't need to cover 1 and 3 anymore, because if we reach $p_{n-1}$ before round $\expo t {n-1}$, we also reach 1 and 3 before round $t'$, by going through $p_{n-1}$.}
    \label{fig:intuition}
\end{figure}

\alphas*

\begin{proof}
    If $s-k$ is odd, let $\ell$ be such that $s-k=2\ell+1$: 
    $$
    \sum_{v:(s,v) \in E}w(s,v)=\sum_{v:(s,v) \in E}2v-k-s
    =\sum_{v:v\leq s\leq 2v-k}2v-k-s
    =\sum_{(s+k)/2\leq v\leq s}2v-k-s
    $$

but then $s-k \equiv s-k+2k \mod 2$, so $s+k$ is also odd, and $\ceil{\frac {s+k} 2}=\ell+k+1$, therefore:

\begin{multline*}
    \sum_{v:(s,v) \in E}w(s,v)=\sum_{\ell+k+1\leq v\leq 2\ell+k+1}2v-k-2\ell-k-1
    =\sum_{0\leq v\leq \ell}2v+1\\=2\frac{\ell(\ell+1)} 2+\ell+1=\frac{s-k-1}{2}\frac{s-k+1}{2}+\frac{s-k+1} 2
    \end{multline*}

        If $s-k$ is even, let $\ell$ be such that $s-k=2\ell$: 
    $$
    \sum_{v:(s,v) \in E}w(s,v)=\sum_{v:(s,v) \in E}2v-k-s
    =\sum_{v:v\leq s\leq 2v-k}2v-k-s
    =\sum_{(s+k)/2\leq v\leq s}2v-k-s
    $$

Therefore:

$$
    \sum_{v:(s,v) \in E}w(s,v)=\sum_{\ell+k\leq v\leq 2\ell+k}2v-k-2\ell-k
    =\sum_{0\leq v\leq \ell}2v=2\frac{\ell(\ell+1)} 2=\frac{s-k}{2}\frac{s-k+2}{2}
    $$

\end{proof}
\subsubsection{The Lower Bound}

Our lower bound for this problem is very similar to the lower bound in \figref{fig:low1}. The lower bound specific for this problem can be found in \figref{fig:low2}. This lower bound is essentially ``trapping'' $k-1$ vertices each in its own 1-vertex tree, and repeating the strategy for the lower bound with $i$ vertices for broadcasting on trees on the last tree. Since none of the first $k-1$ vertices can communicate with the others, cover is only achieved when broadcast is achieved on the last tree. Setting $i=n-k+1$, we get a lower bound for the Covering on $k$-forests of $\ceil{\frac {3n-3k}{2} -1}$:

\lowerbthree*

\begin{figure}[ht]
    \centering
    \includegraphics[width=\linewidth]{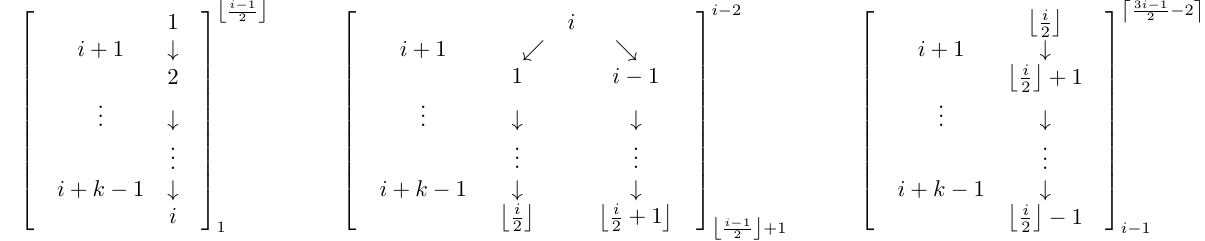}
    \caption{The lower bound example for the Covering on $k$-Forests problem. $\left[ G \right]_a^b$ means that network $G$ is the communication graph for all rounds between $a$ and $b$ (inclusive). Each of the vertices $i+1, \hdots, i+k-1$ is isolated in every round.}
    \label{fig:low2}
\end{figure}

\subsection{\texorpdfstring{$k$-Broadcasting on $k$-Rooted Networks}{k-Broadcasting on k-Rooted Networks}}
\label{appendix1}
\subsubsection{The Lower Bound}

The lower bound for this problem is based on the lower bound for broadcasting on trees, in the \figref{fig:low1}.
\lowerbtwo*
\begin{proof}
The graph for
this lower bound  can be found in \figref{fig:low3}. The idea is  to 
reduce the $k$-broadcasting problem to the broadcast problem.
More specifically, we use the sequence of networks from
the lower bound of \figref{fig:low1} with $i$ vertices for broadcasting on trees, while replacing each of the 3 vertices that act as root in the three networks of \figref{fig:low1} by $k$ fully connected vertices, i.e.~$k$
vertices such that everyone points to the other $k-1$. Then $k$-cover  is only achieved when broadcast is achieved on the original networks. Setting $i=n-3k+3$, we get a lower bound for the $k$-broadcasting on $k$-rooted networks of $\ceil{\frac {3n-9k}{2} +2}$.
\end{proof}

\begin{figure}
    \centering
    \includegraphics[width=0.85\linewidth]{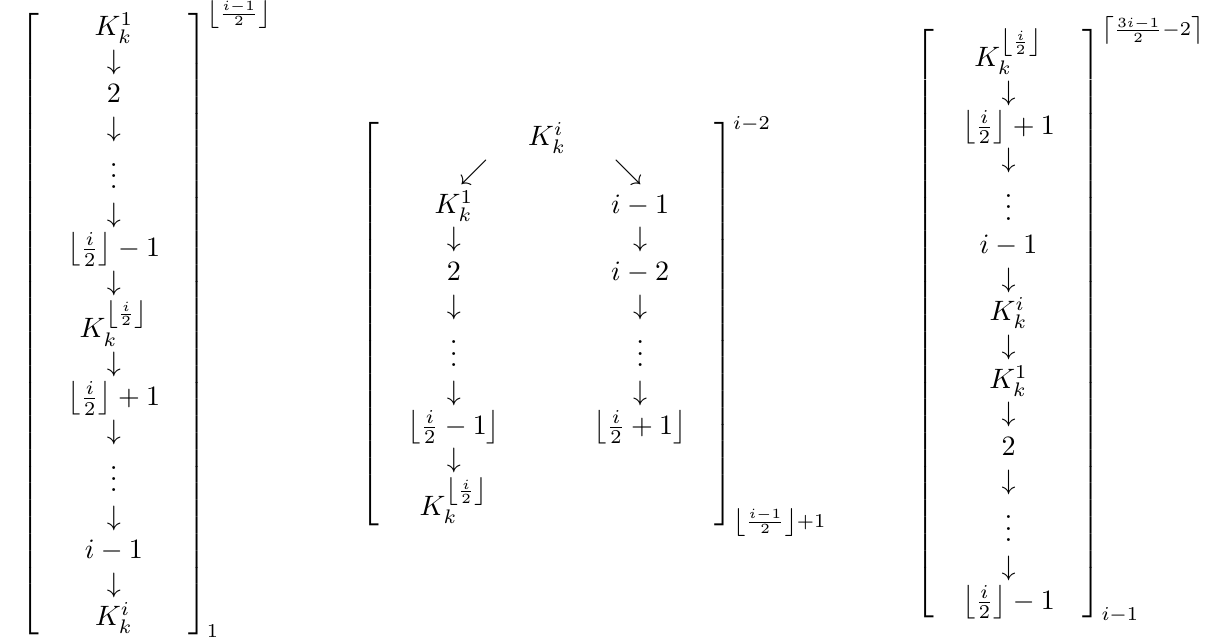}
    \caption{The lower bound example for the Covering on $k$-Forests problem. $\left[ G \right]_a^b$ means that network $G$ is the communication graph for all rounds between $a$ and $b$ (inclusive). $K_k^a$ is the complete graph that replaces vertex $a$ from the example is \figref{fig:low1}. $K_k^a \rightarrow y$ (respectively $y \rightarrow K_k^a$) means that an edge is inserted from every $x \in K_k^a$ to $y$ (respectively from $y$ to every $x \in K_k^a$). Similarly, $K_k^a \rightarrow K_k^b$ means that an edge is inserted from every $x \in K_k^a$ to every $y \in K_k^b$.}
    \label{fig:low3}
\end{figure}

\end{document}